\documentclass[10pt,a4paper]{article}

\usepackage{amsmath}
\usepackage{amssymb}
\usepackage{tikz}
\usepackage{wrapfig}

\newcommand{\igw}[1]{}
\newcommand{\anca}[1]{}

\usepackage{logic9}
\newtheorem{example}[theorem]{Example}

\newcommand{\CAS}{\textsf{CAS}}
\newcommand*\circled[1]{\tikz[baseline=(char.base)]{
            \node[shape=circle,draw,inner sep=1pt] (char) {\small #1};}}
\newcommand{\sgt}{\vartriangleright_\sig}
\newcommand{\run}{\mathit{run}}

\newcommand{\Ccloc}{\Cc^{\dar_{loc}}}
\newcommand{\ts}{\mathit{ts}}
\newcommand{\Cor}{\mathit{Corr}}

\newcommand{\sig}{\mathit{sig}}

\newcommand{\rep}{\mathit{rep}}

\newcommand{\hide}{\mathit{hide}}

\newcommand{\red}[1]{#1^\triangledown}
\newcommand{\short}[1]{#1^\circledS}
\newcommand{\mem}[1]{#1^m}
\newcommand{\slow}{\mathit{slow}}
\newcommand{\Sloc}{\S^{loc}}

\newcommand{\state}{\mathit{state}}

\newcommand{\loc}{\mathit{dom}}

\newcommand{\dar}{\downarrow}

\newcommand{\Ssys}{\S^{sys}}
\newcommand{\Senv}{\S^{env}}

\newcommand{\dom}{\mathit{dom}}

\newcommand{\ch}{\mathop{ch}}

\title{Distributed synthesis for acyclic architectures}

\author{Anca Muscholl\\ Universit\'e de Bordeaux and Igor Walukiewicz\\CNRS, Universit\'e de Bordeaux}%

\begin{document}

\maketitle

\begin{abstract}
  The distributed synthesis problem is about constructing correct
  distributed systems, i.e., systems that satisfy a given
  specification. We consider  a slightly more general problem of
  distributed control, where the goal is to restrict the behavior
  of a given distributed system in order to satisfy the
  specification. Our systems are finite state machines that
  communicate via rendez-vous (Zielonka automata). We show
  decidability of the synthesis problem for all $\omega$-regular local
  specifications, under the restriction that the communication graph
  of the system is acyclic. This result extends a previous
  decidability result for a restricted form of local reachability 
  specifications.

\end{abstract}

\section{Introduction}

\igw{update Intro}
Synthesizing distributed systems from specifications is an attractive
objective, since distributed systems are notoriously difficult to get
right. Unfortunately, there are very few known decidable frameworks
for distributed synthesis. We study a framework for synthesis
of open systems that is based on rendez-vous communication and causal
memory. 
In particular, causal memory implies that specifications can talk about
when a communication takes place, but cannot limit information that is
transmitted during communication. This choice is both realistic and
avoids some pathological reasons for undecidability. We show
a decidability result for acyclic communication graphs and local
$\omega$-regular specifications.

Instead of synthesis we actually work in the more general framework of
distributed control. Our setting is a direct adaptation of the
supervisory control framework of Ramadge and Wonham~\cite{RW89}. In
this framework we are given a plant (a finite automaton) where some of
the actions are uncontrollable, and a specification, and the goal is to
construct a controller (another finite automaton) such that its
product with the plant satisfies the specification. The controller is
not allowed to block uncontrollable actions, in other words, in every state there is a
transition on each uncontrollable action. The controlled plant has
less behaviors, resulting from restricting
controllable actions of the plant.  In our case the formulation is
exactly the same, but we consider Zielonka automata instead of finite
automata, as plants and controllers. Considering parallel devices, as
Zielonka automata, in the standard definition of control gives an
elegant formulation of the distributed control problem.




Zielonka automata~\cite{zie87,ms97} are by now a well-established
model of distributed computation. Such a device is an asynchronous
product of finite-state processes synchronizing on shared
actions. Asynchronicity means that processes can progress at different
speed. The synchronization on shared actions allows the synchronizing
processes to exchange information, in particular the controllers can
transfer control information with each synchronization. This model can
encode some common synchronization primitives available on modern
multi-core processors for implementing concurrent data structures, like compare-and-swap.

We show decidability of the control problem for Zielonka automata
where the communication graph is acyclic: a process can communicate
(synchronize) with its parent and its children. Our specifications are
conjunctions of $\w$-regular specifications for each of the component
processes. We allow uncontrollable communication actions -- the only
restriction is that all communication actions must be
binary. Uncontrollable communications give a big flexibility, for
instance it is possible to model asymmetric situations where
communication can be refused by one partner, but not by the other one.

Our result extends~\cite{GGMW13} that  showed decidability for a
restricted form of local reachability objectives (blocking final
states). We still get the same complexity as in~\cite{GGMW13}:
non-elementary in general, and EXPTIME for architectures of depth $1$.
Covering all $\w$-regular objectives allows to express fairness
constraints but at the same time introduces important technical
obstacles. Indeed, for our construction to work it is essential that
we enrich the framework by uncontrollable synchronization actions. 
This makes a separation into
controllable and uncontrollable states impossible. 
In consequence, we are lead to abandon the game metaphor, to
invent new arguments, and to design a new proof structure.

Most research on distributed synthesis and control has been done in
the setting proposed by Pnueli and Rosner \cite{PR90}. This setting is also based
on shared-variable communication, however it does not allow
to pass additional information between processes. So their model leads
to partial information games, and decidability of synthesis holds
only for very restricted
architectures~\cite{KV01,MadThiag01,FinSch05}. 
While specifications leading to
undecidability are very artificial, no elegant solution to eliminate
them exists at present.
The synthesis setting is investigated
in~\cite{MadThiag01} for local specifications, meaning that each
process has its own, linear-time specification. For such
specifications, it is shown that an architecture has a decidable
synthesis problem if and only if it is a sub-architecture of a
pipeline with inputs at both endpoints. 
More relaxed variants of synthesis have been proposed, where the
specification does not fully describe the communication of the
synthesized system. One approach consists in adding communication in
order to combine local knowledge, as proposed for example
in~\cite{gpq12}.  Another approach is to use specifications only for
describing external communication, as done in~\cite{gs13tocl} on
strongly connected architectures where processes communicate via
signals.

Apart from~\cite{GGMW13}, two closely related decidability results for
synthesis with causal memory are known, both of different flavor than
ours. The first one \cite{GLZ04} restricts the alphabet of actions:
control with reachability condition is decidable for co-graph
alphabets. This restriction excludes among others client-server
architectures, which are captured by our setting. The second result
\cite{MTY05} shows decidability by restricting the plant: roughly
speaking, the restriction says that every process can have only
bounded missing knowledge about the other processes, unless they
diverge (see also \cite{ms13} that shows a doubly exponential upper
bound). The proof of \cite{MTY05} goes beyond the controller synthesis
problem, by coding it into monadic second-order theory of event
structures and showing that this theory is decidable when the
criterion on the plant holds. Unfortunately, very simple plants have a
decidable control problem but undecidable MSO-theory of the associated
event structure. Game semantics and
asynchronous games played on event structures are considered in~\cite{mel06}. More recent work
\cite{gw13} considers games on event structures and shows a Borel
determinacy result for such games under certain restrictions. 

\emph{Overview.} In Section~\ref{sec:defs} we state our control
problem, and in Section~\ref{sec:reduction} we give the main lines of
the proof, that works by a reduction of the number of processes. In
Section~\ref{sec:short} we show that we may assume for the process
that is eliminated that there is a bound on the number of local actions
it can perform between consecutive synchronizations with its
parent. In Section~\ref{sec:new} we present the reduction, and in
Sections~\ref{sec:C}, \ref{sec:D} we show the correctness of the
construction.

\section{Control for Zielonka automata}\label{sec:defs}

In this section we introduce our control problem for Zielonka
automata, adapting  the definition of
supervisory control~\cite{RW89} to our model. 


A Zielonka automaton~\cite{zie87,ms97} 
is a simple
distributed finite-state devices. Such an automaton is a parallel
composition of several finite automata, called~\emph{processes},
synchronizing on shared actions. There is no global clock, so between
two synchronizations, two processes can do a different number of
actions. Because of this, Zielonka automata are also called
asynchronous automata.

A \emph{distributed action alphabet} on a finite set $\PP$ of processes is a
pair $(\S,\loc)$, where $\S$ is a finite set of \emph{actions} and
$\loc:\S \to (2^{\PP}\setminus \es)$ is a \emph{location
  function}. The location $\loc(a)$ of action $a \in\S$ comprises all
processes  that need to synchronize in order to perform this
action. Actions from $\S_p=\set{a \in\S \mid p
\in\loc(a)}$ are called \emph{$p$-actions}. 
We write $\Sloc_p=\set{a \mid \dom(a)=\set{p}}$
for the set of~\emph{local} actions of $p$.

A (deterministic) \emph{Zielonka automaton}
$\Aa=\struct{\set{S_p}_{p\in \PP},s_{in},\set{\d_a}_{a\in\S}}$ is
given by:
\begin{itemize}
\item for every process $p$ a finite set $S_p$ of (local) states,
\item the initial state $s_{in} \in \prod_{p \in \PP} S_p$, 
\item for every action $a \in\S$ a partial transition function
  $\d_a:\prod_{p\in \loc(a)}S_p\stackrel{\cdot}{\to} \prod_{p\in
    \loc(a)}S_p$ on tuples of states of processes in $\loc(a)$.
\end{itemize}

\begin{example}\label{ex:cas}
  Boolean multi-threaded programs with shared
  variables can be modeled as Zielonka automata. As an example we
  describe the translation for the \emph{compare-and-swap} (CAS)
  instruction. This instruction has $3$ parameters: \textsf{CAS}($x$:
  variable; \emph{old}, \emph{new}: int). Its effect is to return the
  value of $x$ and at the same time set the value of $x$ to
  \emph{new}, but only if the previous value of $x$ was equal to
  \emph{old}. The compare-and-swap operation is a  widely used primitive
  in implementations of concurrent data structures, and has
  hardware support in most contemporary multiprocessor
  architectures.

  Suppose that we have a thread $t$, and a shared variable $x$ that is
  accessed by a CAS operation in $t$ via $y:=\CAS_x(i,k)$. So $y$ is a
  local variable of $t$. In the Zielonka automaton we will have one
  process modeling thread $t$ and one process for variable $x$. The
  states of $t$ will be valuations of local variables. The states of
  $x$ will be the values $x$ can take. The $\CAS$ instruction above
  becomes a synchronization action. We have the following two types of
  transitions on this action:
  \begin{center}
    \includegraphics[scale=.4]{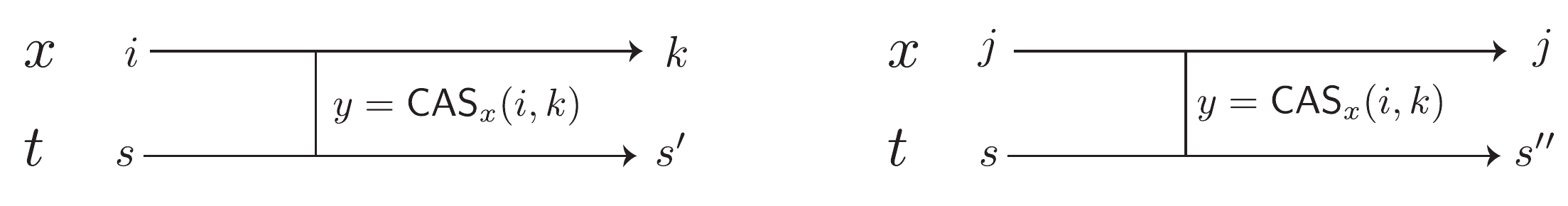}
  \end{center}
Notice that in state $s'$, we have $y=i$, whereas in $s''$, we have $y=j$.
\end{example}

For convenience, we abbreviate a tuple $(s_p)_{p \in P}$ of local
states by  $s_P$,  where $P \subseteq \PP$. We also talk about $S_p$
as the set of \emph{$p$-states}.

A Zielonka automaton can be seen as a sequential automaton with the
state set $S = \prod_{p\in\PP} S_p$ and transitions $s \act{a} s'$ if
$(s_{\loc(a)}, s'_{\loc(a)}) \in \d_a$, and $s_{\PP\setminus
  \loc(a)}=s'_{\PP\setminus \loc(a)}$. So the states of this automaton
are the tuples of states of the processes of the Zielonka
automaton. For a process $p$ we will talk about the
\emph{$p$-component} of the state.  A run of $\Aa$ is a finite or
infinite sequence of transitions starting in $s_{in}$. Since the
automaton is deterministic, a run is determined by the sequence of
labels of the transitions. We will write $\run(u)$ for the run
determined by the sequence $u\in \S^\infty$. Observe that $\run(u)$
may be undefined since the transition function of $\Aa$ is partial. We
will also talk about the projection of the run on component $p$,
denoted $\run_p(u)$, that is the projection on component $p$ of the
subsequence of the run containing the transitions involving $p$. We
will assume that every local state of $\Aa$ occurs in some run. For
finite $w$ let $\state(w)$ be the last state in $\run(w)$. By
$\loc(u)$ we denote the union of $\loc(a)$, for all $a \in\S$
occurring in $u$.

We will be interested in maximal runs of Zielonka
automata. For parallel devices the notion of a maximal run is not that
evident, as one may want to impose some fairness conditions. We settle here
for a minimal sensible fairness requirement. It says that a run is
maximal if processes that have only finitely many actions in the run
cannot perform any additional action.

\begin{definition}[Maximal run]
For a word $w \in\S^\infty$ such that $\run(w)$ is
defined, we say that $\run(w)$ is
\emph{maximal} if there is no 
decomposition $w=uv$, and no action
$a \in \S$ such that $\loc(v) \cap \loc(a)=\es$ and $\run(uav)$ is defined. 
\end{definition}


Automata can be equipped with a \emph{correctness condition}. We prefer
to talk about correctness condition rather than acceptance condition
since we will be interested in the set of runs of an automaton rather
than in the set of words it accepts. We will consider local
regular correctness conditions: every process has its own correctness
condition $\Cor_p$. A run of $\Aa$ is
\emph{correct} if for every process $p$, the projection of the run on
the transitions of $\Aa_p$ is in $\Cor_p$. Condition
$\Cor_p$ is 
specified by a set $T_p \incl {S_p}$ of terminal states and an
$\omega$-regular set $\Omega_p \subseteq (S_p \times \S_p \times
S_p)^\omega$. 
A sequence $(s^0_p,a_0,s^1_p) (s^1_p,a_1,s^2_p) \dots$ satisfies
$\Cor_p$ if either: (i) 
it is finite and ends with a state from $T_p$, or (ii)
it is infinite and belongs to $\Omega_p$.
At this stage the set of terminal states $T_p$ may look unnecessary, but
it will simplify our constructions later. 

Finally, we will need the notion of \emph{synchronized product}
$\Aa\times\Cc$ of two Zielonka automata. For
$\Aa=\struct{\set{S_p}_{p\in \PP},s_{in},\set{\d^A_a}_{a\in\S}}$ and
$\Cc=\struct{\set{C_p}_{p\in \PP},c_{in},\set{\d^C_a}_{a\in\S}}$ let
$\Aa \times \Cc=\struct{\set{S_p\times C_p}_{p\in
    \PP},(s_{in},c_{in}),\set{\d^\times_a)_{a\in\S}}}$ where there is
a transition from $(s_{\loc(a)},c_{\loc(a)})$ to
$(s'_{\loc(a)},c'_{\loc(a)})$ in $\d^\times_a$ iff $(s_{\loc(a)},
s'_{\loc(a)}) \in \d^A_a$ and $(c_{\loc(a)}, c'_{\loc(a)}) \in \d^C_a$.

\medskip

To define the control problem for Zielonka automata we fix a
distributed alphabet $\struct{\PP,\loc:\S\to (2^\PP \setminus
  \es)}$. We partition $\S$ into the set of \emph{system actions} $\Ssys$ and
\emph{environment actions} $\Senv$. Below we will introduce the notion of
controller, and require that it does not block environment
actions. For this reason we speak about
\emph{controllable/uncontrollable} actions when referring to system/environment
actions. We impose three simplifying assumptions: (1) All actions are
at most binary ($|\loc(a)|\le 2$ for every $a \in\S$); (2) every
process has some controllable action; (3) 
all controllable actions are local. 
Among the three conditions only the first one is indeed a restriction of
our setting. The other two are not true limitations, in particular controllable
shared actions can be simulated by a local controllable choice, followed by 
non-controllable local or shared actions (see
Proposition~\ref{prop:comm_controllable}).

\begin{definition}[Controller, Correct Controller]
  A \emph{controller} is a Zielonka automaton that cannot block
  environment (uncontrollable) actions. In other words, from every
state every environment action is possible: for every
  $b\in\Senv$, $\d_b$ is a total function.  We say that a controller
  $\Cc$ \emph{is correct for} a plant $\Aa$ if all maximal runs of $\Aa\times
  \Cc$ satisfy the correctness condition of $\Aa$.
\end{definition}

Recall that an action is possible in  $\Aa\times
  \Cc$ iff it is possible in both $\Aa$ and $\Cc$. By the above
  definition, environment actions
  are always possible in $\Cc$. The major difference between the
  controlled system $\Aa \times \Cc$ and
and $\Aa$ is that the states of $\Aa \times \Cc$ carry  the
additional information computed by $\Cc$, and that $\Aa \times
  \Cc$ may have less behaviors, resulting from disallowing controllable actions
  by $\Cc$.

The correctness of $\Cc$ means that all the runs of $\Aa$ that are
\emph{allowed} by $\Cc$ are correct. In particular, $\Cc$ does not have
a correctness condition by itself. Considering only maximal runs of
$\Aa \times \Cc$ imposes some minimal fairness conditions: for example
it implies that if a process can do a local action almost
always, then it will eventually do some action. 

\begin{definition}[Control problem]
  Given a distributed alphabet $\struct{\PP,\loc:\S\to (2^\PP \setminus
    \es)}$ together with a partition of actions $(\Ssys,\Senv)$, and
  given a Zielonka automaton $\Aa$ over this alphabet, 
 find a controller $\Cc$ over the same alphabet such that $\Cc$ is
  correct for $\Aa$.
\end{definition}

The important point in our definition is that the controller has
the same distributed alphabet as the automaton it controls, in other
words the controller is  not allowed to introduce additional
synchronizations between processes.  

\begin{example}\label{ex:control}
  We give an example showing how causal memory works and helps to
  construct controllers. Consider an automaton $\Aa$ with $3$ processes:
  $p$, $q$, $r$. We would like to control it so that the only two
  possible runs of $\Aa$ are the following:
  \begin{center}
    \includegraphics[scale=.6]{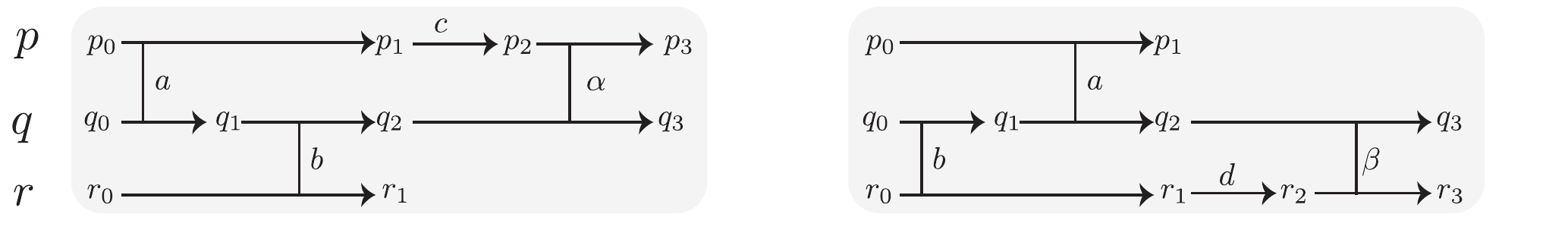}
  \end{center}
 So $p$ and $q$ should synchronize on $\a$ when action $a$ happened
 before $b$, otherwise $q$ and $r$ should synchronize on
 $\b$. Communication actions are uncontrollable, but the
 transitions of $\Aa$ are such that there are local controllable
 actions $c$ and $d$ that enable communication on $\a$ and $\b$
 respectively. So the controller should block either $c$ or $d$
 depending on the order between $a$ and $b$. The transitions of $\Aa$
 are as follows
 \begin{align*}
   \d_a(p_0,q_0)=& (p_1,q_1) & \d_a(p_0,q_1)=& (p_1,q_2) &
\d_b(q_0,r_0)=&(q_1,r_1) & \d_b(q_1,r_0)=&(q_2,r_1)\\
\d_c(p_1)=&p_2 & \d_\a(p_2,q_2)=& (p_3,q_3) & \d_d(r_1)=&r_2  & \d_\b(q_2,r_2)=&(q_3,r_3)
 \end{align*}
 These transitions allow the two behaviors depicted above but also
 two unwanted ones, as say, when $a$ happens before $b$ and then we
 see $\b$.  Clearly, the specification of the desired behaviors can be
 formulated as a local 
 condition on $q$. So by encoding some information in states of $q$
 this condition can be expressed by a set of terminal states $T_q$. We
 will not do this for readability.

 The controller $\Cc$ for $\Aa$ will mimic the structure of $\Aa$: for
 every state of $\Aa$ there will be in $\Cc$ a state with over-line. So,
 for example, the states of $q$ in $\Cc$ will be $\bar q_0, \dots,
 \bar q_3$. Moreover $\Cc$ will have two new states $\underline{p_1}$
 and $\underline{r_1}$. The transitions will be
 \begin{align*}
   \d_a(\bar p_0, \bar q_0)=& (\bar p_1,\bar q_1) & 
   \d_a(\bar p_0,\bar q_1)=& (\underline{p_1},\bar q_2) &
\d_b(\bar q_0,\bar r_0)=&(\bar q_1,\bar r_1) & 
\d_b(\bar q_1,\bar r_0)=&(\bar q_2,\underline{r_1})\\
\d_c(\bar p_1)=&\bar p_2 & \d_c(\underline{p_1})=&\bot & \d_d(\bar r_1)=&\bar
r_2 & \d_d(\underline{r_1})=&\bot\\
\d_\a(\bar p_2,\bar q_2)=& (\bar p_3,\bar q_3) & && \d_\b(\bar q_2,\bar
r_2)=&(\bar q_3,\bar r_3)
 \end{align*}
Observe that $c$ is blocked in $\underline{p_1}$, and so is $d$ from $\underline{r_1}$. It is easy to verify that the runs of $\Aa\times \Cc$ are as
required, so $\Cc$ is a correct controller for $\Aa$. (Actually the definition of a controller forces us to make
 transitions of $\Cc$ total on uncontrollable actions. We can do it in
 arbitrary way as this will not add new behaviors to $\Aa\times\Cc$.)

This example shows several phenomena. The states of $\Cc$ are the
states of $\Aa$ coupled with some additional information. 
We formalize this later under a notion of covering controller.
We could also see above a case where a communication is decided by one
of the parties. Processes $p$, thanks to a local action,
can decide if it wants to communicate via $\a$, but process $q$ has
to accept $\a$ always. This shows the flexibility given by
uncontrollable communication actions. Finally, we could see
information passing during communication. In $\Cc$ process $q$
passes to $p$ and $r$ information about its local state (transitions
on $a$ and on $b$).
\end{example}

We end the section by showing the assumption that controllable
actions are local, is not a restriction.

\begin{proposition}\label{prop:comm_controllable}
  The control problem for Zielonka automata where communication
  actions may be controllable, reduces to the setting where
  controllable actions are all local.
\end{proposition}

\begin{proof}
  We start with an automaton $\Aa$ over a distributed alphabet
  $\struct{\S,\loc}$ and a correct covering controller $\Cc$.  We
  define first a new automaton $\Aa'$ over an extended distributed
  alphabet $\struct{\S',\loc'}$ with $\S'=\S \cup \set{\ch(A) \mid A
    \subseteq \Ssys_p \text{ for some } p \in \PP}$. All new actions
  are local: $\loc'(\ch(A))=\set{p}$ if $A \subseteq \Ssys_p$; the
  domain of other actions do not change. What changes is that all old
  actions become uncontrollable, and the only controllable actions in
  $\S'$ are those of the form $\ch(A)$.

  \begin{itemize}
  \item The set of $p$-states of $\Aa'$ is the set of $p$-states of $\Aa$,
    plus some new states of the form $\struct{s_p,A}$ where $s_p$ is a
    $p$-state of $\Aa$ and $A \subseteq \Ssys_p$. 
\item For every old $p$-state $s_p$ we delete all outgoing
  controllable transitions and add
  \[s_p \act{\ch(A)} \struct{s_p,A}\,,\] for every set $A$ of
  controllable actions enabled in $s_p$. From $\struct{s_p,A}$ we put in $\Aa'$ 
  transitions as follows. If $a \in A$ is local then we
  have $\struct{s_p,A} \act{a} s'_p$ whenever
  $s_p \act{a} s'_p$ in $\Aa$. If $a \in A \cap B$ and
  $\loc(a)=\set{p,p'}$ then we have $(\struct{s_p,A},\struct{s_{p'},B})
  \act{a} (s'_p,s'_{p'})$ whenever $(s_p,s_{p'}) \act{a}
  (s'_p,s'_{p'})$ in $\Aa$.
\item The correctness condition of $\Aa'$ is a straightforward
  modification of the one of $\Aa$. 

  \end{itemize}

  Assume first that $\Cc'$ is a correct covering controller for
  $\Aa'$. From $\Cc'$ we define the automaton $\Cc$ over the same sets
  of states, by modifying slightly the transitions as follows. Suppose
  that $c \act{\ch(A)} d$ is a (local) transition in $\Cc'_p$. Since
  $\Cc'$ is covering we have a transition of the form $s_p=\pi'(c)
  \act{\ch(A)} \pi'(d)=\struct{s_p,A}$ in $\Aa'$. Let $a \in A$ be
  local. Since $a$ is uncontrollable in $\Aa'$ and $\struct{s_p,A}
  \act{a} s'_p$ (for some $s'_p$) we must also have $d \act{a} e$ for
  some state $e$ of $\Cc'_p$, since $\Cc'$ is covering. We delete $c
  \act{\ch(A)} d$ from $\Cc'$ and replace $d \act{a} e$ by $c \act{a}
  e$. If $a$ is shared by $p,p'$, let us consider some transition $c'
  \act{\ch(B)} d'$ with $a \in B$ in $\Cc'_{p'}$. Since $a$ is
  uncontrollable in $\Aa'$ we find again some transition
  $(d,d')\act{a} (e,e')$ in $\Cc'$. We replace then $(d,d')\act{a}
  (e,e')$ by $(c,c') \act{a} (e,e')$ in $\Cc$. Of course, this is done
  in parallel for all transitions labeled by some $\ch(A)$. It is
  immediate that $\Cc$ is covering $\Aa$, by taking
  $\pi=\pi'$. Maximal runs of $\Cc$ map to maximal runs of $\Cc'$ and
  thus satisfy the correctness condition for $\Aa$.

Conversely, given a correct covering controller $\Cc$ for $\Aa$ we
define $\Cc'$ for $\Aa'$. Local $p$-states of $\Cc'$ are
those of $\Cc$, plus additional states of the form $c_A$, where $c$ is
a $p$-state of $\Cc$ and $A \subseteq \Ssys_p$. Consider any $p$-state
$c$ of $\Cc$, and let $A$ be the set of controllable actions enabled
in $c$ (a communication action $a$ with $\loc(a)=\set{p,p'}$ is enabled
in $c$ if there exists some $p'$-state $c'$ and an $a$-transition from
$(c,c')$). We replace all controllable transitions from $c$ by one
(local) controllable transition $c \act{\ch(A)} c_A$, plus some uncontrollable
transitions. If $a\in A$ is local, then we add the 
uncontrollable transitions $c_A \act{a} d$ whenever  $c
\act{a} d$ in $\Cc$. If $\loc(a)=\set{p,p'}$, $(c,c') \act{a}
(d,d')$ in $\Cc$, and $B$ is the set of controllable actions
enabled in the $p'$-state $c'$,  then we replace $(c,c') \act{a}
(d,d')$ by $(c_A,c'_B) \act{a} (d,d')$. Extending $\pi$ by $\pi'(c_A)=\struct{\pi(c),A}$
shows that $\Cc'$ is a covering controller for $\Aa'$. Maximal runs of $\Cc'$ satisfy the
acceptance condition as for $\Cc$.
\end{proof}

\section{Decidability  for acyclic architectures}\label{sec:reduction}
In this section we present the main result of the paper.  We show the
decidability of the control problem for Zielonka automata with acyclic
architecture. A \emph{communication architecture} of a distributed
alphabet is a graph where nodes are processes and edges link processes
that have common actions. An \emph{acyclic architecture} is one whose
communication graph is acyclic\igw{changed}. For example, the communication graph
of the alphabet from the example on
page~\pageref{ex:control}\igw{reference} is a tree with the root $q$
and two successors, $p$ and $r$.

\begin{theorem}\label{thm:main}
  The control problem for Zielonka automata over distributed alphabets
  with acyclic architecture is decidable. If a controller exists, then
  it can be effectively constructed.
\end{theorem}

The remaining of this section is devoted to the outline of the proof
of Theorem~\ref{thm:main}. This proof works by induction on the number
$|\PP|$ of processes in the automaton. A Zielonka automaton over a
single process is just a finite automaton, and the control problem is
then just the standard control problem as considered by Ramadge and
Wonham but extended to all $\w$-regular conditions~\cite{AVW02}. If
there are several processes that do not communicate, then we can solve
the problem for each process separately.

 \begin{figure}[htb]
\centerline{ \includegraphics[scale=.4]{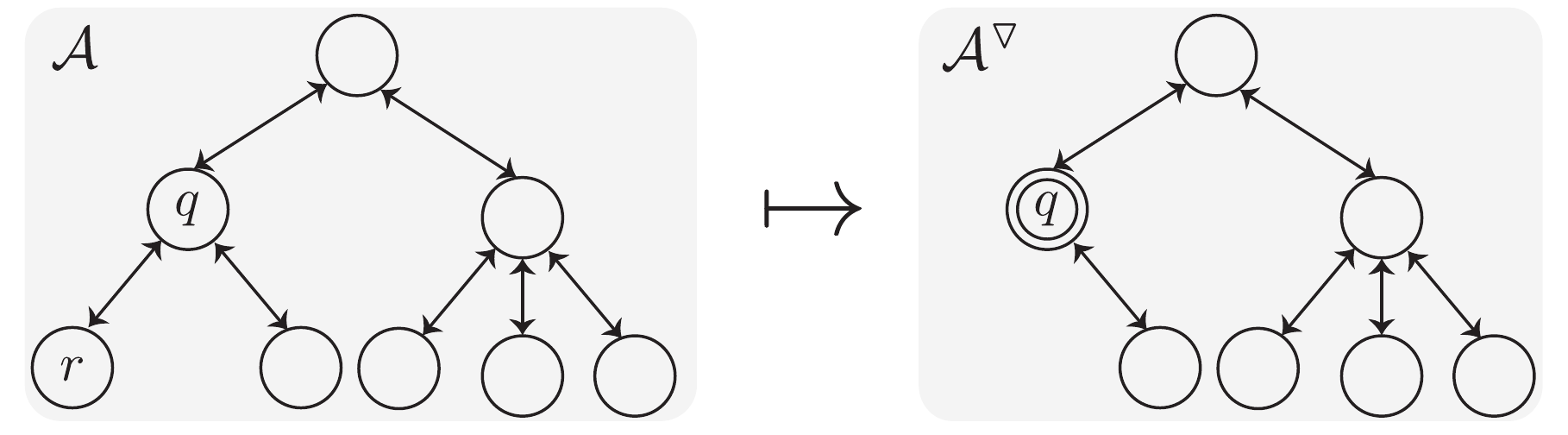}}
\caption{Eliminating process $r$: $r$ is glued with $q$.}   
   \label{fig:schema}
 \end{figure}

Otherwise we choose a leaf process
$r$ and its parent $q$, and construct a new plant $\red\Aa$ over $\PP
\setminus \set{r}$.We will show that the control
problem for $\Aa$ 
has a solution iff the one for $\red\Aa$ does. Moreover, for
every solution for $\red\Aa$ we will be able to construct a solution
for $\Aa$. 

For the rest of this section let us fix the distributed alphabet
$\struct{\PP,\loc:\S\to (2^\PP \setminus \es)}$, 
the leaf process $r$ and
its parent $q$, and  a Zielonka automaton with a correctness
condition $\Aa=\struct{\set{S_p}_{p\in \PP},s_{in},\set{\d_a}_{a\in\S},
\set{\Cor_p}_{p\in \PP}}$.

The first step in proving Theorem~\ref{thm:main} is to simplify the
problem. First,  we can restrict to controllers of a special form
called covering controllers. Next, we show that the component of $\Aa$
to be eliminated, that is $\Aa_r$, can be assumed to have a particular property
($r$-short). After these preparatory results we will be able to present
the reduction of $\Aa$ to $\red \Aa$ (Section~\ref{sec:new}).

\subsection{Covering controllers}

The notion of a covering controller will simplify the presentation
because it will allow us to focus on the runs of the controller
instead of a product of the plant and the controller.\igw{refer to the
  example}

\begin{definition}[Covering controller]~\label{df:covering controller}
  Let $\Cc$ be a Zielonka automaton over the same alphabet as $\Aa$;
  let $C_p$ be the set of states of process $p$ in $\Cc$. 
  Automaton $\Cc$ is a \emph{covering controller} for $\Aa$ if
  there is a function $\pi :\set{C_p}_{p\in \PP} \to \set{S_p}_{p\in
    \PP}$, mapping each $C_p$ to $S_p$ and satisfying two conditions:
  (i)  if $c_{\loc(b)} \act{b} c'_{\loc(b)}$ then
  $\pi(c_{\loc(b)})\act{b}\pi(c'_{\loc(b)})$; (ii) for
  every uncontrollable action $a$: if $a$ is enabled from
  $\pi(c_{\loc(a)})$ then it is also enabled from $c_{\loc(a)}$.
\end{definition}

\begin{remark} 
Strictly speaking, a covering controller $\Cc$ may not be a
controller since we  do not require that every uncontrollable action
is enabled in every state, but only those actions that are enabled
in $\Aa$. From $\Cc$ one can get a controller $\hat{\Cc}$ by adding
self-loops for all missing uncontrollable transitions. 

Notice that thanks to the projection $\pi$, a covering controller can
inherit the correctness condition of $\Aa$. Moreover, the sequences
labeling the maximal runs
of $\Cc$, $\Aa \times \Cc$ and $\Aa \times \hat{\Cc}$ are the same.
\end{remark}

\begin{lemma}\label{lemma:covering}
  There is a correct controller for $\Aa$ if and only if there is a
  covering controller $\Cc$ for $\Aa$ such that all the maximal runs
  of $\Cc$ satisfy the inherited correctness condition.
\end{lemma}

\begin{proof}
  If $\Cc$ is a covering
controller for $\Aa$ such that all its 
  maximal runs satisfy the inherited correctness condition then
  $\hat{\Cc}$ is a correct controller for $\Aa$. 
Conversely, if $\Cc$ is a correct controller for $\Aa$ then
  $\Aa\times\Cc$ is a covering controller where all maximal runs
  satisfy the inherited correctness condition.
\end{proof}

We will refer to a covering controller with the
property that  all its
maximal runs satisfy the inherited correctness condition, as 
\emph{correct covering controller}.

\subsection{Short automata}\label{sec:short}

In this section we justify our restriction to plants $\Aa$ where the
$r$-component $\Aa_r$ is short (see Definition~\ref{def:r-short}
below).  Recall that we have assumed that all controllable actions are
local and that we consider a tree architecture with a leaf process $r$
and its parent $q$.

\begin{definition}[$r$-short]\label{def:r-short}
 Automaton $\Aa$ is \emph{$r$-short} if there
  is a bound on the number of actions that $r$ can perform without
  doing a communication with $q$. 
\end{definition}

\begin{theorem}\label{thm:short}
  For every automaton $\Aa$, we can construct an $r$-short automaton
  $\short\Aa$ such that there is a correct controller for $\Aa$ iff
  there is one for $\short\Aa$.
\end{theorem}
The rest of this subsection is devoted to the proof of the above theorem.
Theorem~\ref{thm:short} bears some resemblance with the fact that
every parity game can be transformed into a finite game: when a loop
is closed the winner is decided looking at the ranks on the loop. This
construction would do if $r$ had no interaction with $q$. Possible
interactions with $q$ make the construction more involved. Moreover,
need to prove existence of some kind of memoryless strategies for
distributed controllers.

Observe that we can make two simplifying assumptions. First, we assume
that the correctness condition on $r$ is a parity condition. That is,
it is given by a rank function $\W_r:S_r\to \Nat$ and the set of
terminal states $T_r$. We can assume this since every regular language
of infinite sequences can be recognized by a deterministic parity
automaton. 
The second simplification is to assume that the automaton $\Aa$ is
\emph{$r$-aware} with respect to the parity condition on $r$. This
means that the state of $r$ determines the biggest rank that
has been seen since the last communication of $r$ with $q$. It is easy
to transform an automaton to an $r$-aware one.

Recall that if $\Cc$ is a covering controller for $\Aa$ (cf.\
Definition~\ref{df:covering controller}) then
there is a function $\pi :\set{C_p}_{p\in
  \PP} \to \set{S_p}_{p\in \PP}$, mapping each $C_p$ to $S_p$ and respecting
the transition relation: if $c_{\loc(b)} \act{b} c'_{\loc(b)}$ then
$\pi(c_{\loc(b)})\act{b}\pi(c'_{\loc(b)})$. 

\begin{definition}[$r$-memoryless controller]
  A covering controller $\Cc$ for $\Aa$ is \emph{$r$-memoryless} when
  for every pair of states $c_r\not=c'_r$ of $\Cc_r$: 
if there
  is a path on local $r$-actions from $c_r$ to $c'_r$ then
  $\pi(c_r)\not=\pi(c'_r)$.
\end{definition}

Intuitively, a controller can be seen as a strategy, and $r$-memoryless
means that it does not allow the controlled automaton to go twice
through the same $r$-state between two consecutive communication
actions of $r$ and $q$.

\begin{lemma}\label{lemma-r-memoryless}
  Fix an $r$-aware automaton $\Aa$ with a parity correctness condition
  for process $r$. 
  If there is a correct controller for $\Aa$ then there is
  also one that is covering and $r$-memoryless.
\end{lemma}


The proof of Lemma~\ref{lemma-r-memoryless} uses the notion of
  signatures, that is classical in 2-player parity games, for defining
  a $r$-memoryless controller $\Cc^m$ from $\Cc$. The idea is to use
  representative states of $C_r$, defined in each strongly connected
  component according to a given signature and covering function
  $\pi$.

By Lemma~\ref{lemma:covering} we can assume that we have a covering
controller for $\Aa$. 
Let us fix an arbitrary
linear order on the set $C_r$ of states of the automaton $\Cc_r$. Let
$\Ccloc_r$ denote the graph obtained from $\Cc_r$ by taking $C_r$ as
set of vertices and the transitions on local $r$-actions as
edges. Since $\Cc$ is a covering controller, every sequence of actions
in $\Cc$ can be performed in the controlled plant. Since $\Cc$ is
correct for $\Aa$, every
infinite sequence of local $r$-actions in the controlled plant
satisfies the parity condition. We can lift this parity condition
directly to $\Cc$ thanks to the fact that $\Cc$ is covering. We obtain
that every infinite path in $\Ccloc_r$ satisfies the parity condition.

  Before proceeding it will be convenient to recall some facts about
  parity games, in particular the notion of
  signature (or progress measure)~\cite{wal01ic}. We consider $\Ccloc_r$ as a
  parity game. Suppose that it uses priorities from
  $\set{1,\dots,d}$. A signature is a $d$-tuple of natural numbers,
  that is, an element of $\Nat^d$. We will be interested in
  assignments of signatures to states of $\Ccloc_r$, that is in
  functions $\sig:C_r\to\Nat^d$. Signatures are ordered
  lexicographically. We write $\sig(c)\geq \sig(c')$ if the signature
  assigned to $c$ is lexicographically bigger or equal to that of
  $c'$. For $i\in\set{1,\dots,d}$ we write $\sig(c)\geq_i \sig(c')$ if
  the signature of $c$ truncated to the first $i$ positions is
  lexicographically bigger or equal to the signature of $c'$ truncated
  to the first $i$ positions.  For a fixed assignment of signatures
  $\sig$ and two states $c$, $c'$ of $\Cc_r$ we write $c\sgt c'$ if
  \begin{equation*}
    \sig(c)\geq_{\W(c)} \sig(c')\quad\text{and the inequality is strict
      if $\W(c)$ is odd.}
  \end{equation*}
  We say that an assignment of signatures $\sig:C_r\to\Nat^d$ is
  \emph{consistent} if for every edge $(c,c')$ of $\Ccloc_r$ we have
  $c\sgt c'$. We now recall a fact that holds for every finite
  parity game, but we specialize them to $\Ccloc_r$.

 \noindent
\textbf{Fact.} 
 Every path of $\Ccloc_r$ satisfies the parity
  condition iff there is a consistent assignment of signatures to
  states of $\Ccloc_r$.

  After these preparations we can define for every state $c_r$ of
  $C_r$ its
  representative state in $C_r$, denoted $\rep(c_r)$, as the unique
  state $c'_r$ satisfying the following conditions:
  \begin{enumerate}
\item $\pi(c_r) = \pi(c'_r)$ and $c'_r$ is reachable from $c_r$;
\item for every $c''_r$ with $\pi(c_r)=\pi(c''_r)$: if $c''_r$ is
  reachable from $c'_r$ then it belongs to the same SCC as
  $c'_r$;
\item among all states satisfying points (1) and (2) consider those
  with the smallest signature; if there is more than one such state
  then pick the state that is the smallest in our fixed arbitrary ordering.
  \end{enumerate}

  \begin{remark}\label{rem:rep} 
    For every $c'_r$ reachable in $\Ccloc_r$ from $\rep(c_r)$: if
    $\pi(c'_r)=\pi(c_r)$ then $\rep(c'_r)=\rep(c_r)$. Indeed, by
    conditions (1) and (2) above $\rep(c'_r)$ and $\rep(c_r)$ must be
    in the same SCC. But then, the representative is uniquely
    determined by signature and ordering. 
  \end{remark}

  We define now $\mem\Cc_r$ from $\Cc_r$ by redirecting every
  transition on a local $r$-action to representatives: if the transition
  goes to a state $c_r$ we make it go to $\rep(c_r)$. Of course,
  $\mem\Cc$ is still covering and
  the above remark 
implies that it is $r$-memoryless.

  \begin{remark}\label{rem:path}
    If we have a transition $c_r\act{b} c'_r$ in $\mem\Cc_r$ then
    there is a sequence $z \in \S^{loc}_r$ of local $r$-actions and
    some state $c''_r$ such that
    $c_r\act{b} c''_r\act{z} c'_r$ in $\Cc_r$.
  \end{remark}

  Remark~\ref{rem:path} allows to map paths in $\mem\Cc_r$ into paths
  in $\Cc_r$. Consider a state $c_1$ of $\mem\Cc_r$ and a finite
  sequence $x\in(\Sloc_r)^*$ such that $x=b_1\cdots b_k$ labels some path
  from $c_1$ in $\mem\Cc_r$, say $u=c_1\act{b_1} c_2\act{b_2}
  c_3\cdots \act{b_k}c_{k+1}$. Remark~\ref{rem:path} gives us a
  sequence $\rep^{-1}(c_1,x)=b_1z_1b_2z_2\dots b_kz_k$, and a
  corresponding path in $\Cc_r$: $c_1\act{b_1z_1}
  c_2\act{b_2z_2}c_3\dots \act{b_kz_k} c_{k+1}$, for some $z_i\in
  (\Sloc_r)^*$. In particular the two paths end in the same state. Of
  course $\rep^{-1}(c_1,x)$ is defined similarly for infinite
  sequences $x$.

\medskip

\textbf{Proof of Lemma~\ref{lemma-r-memoryless}.}
  We are ready to show that $\mem\Cc$ obtained from $\Cc$ by replacing $\Cc_r$
  with $\mem\Cc_r$ satisfies the parity condition. For this take a
  maximal run and suppose towards a contradiction
  that it does not satisfy the parity condition.

  If on this run there are infinitely many communications between $q$ and $r$ 
  then there is an equivalent run whose labeling has the form:
  \begin{equation}
    u=y_0x_0a_1y_1x_1a_2\dots
  \end{equation}
  where $a_i\in \S_q \cap \S_r$, $x_i\in (\S^{loc}_r)^*$, and $y_i\in
  (\S\setminus\S_r)^*$. Here two runs are equivalent means that the
  projections of the two runs on every process are identical. In
  particular, if two runs are equivalent and one of them satisfies the
  correctness condition then so does the other. 

  Let $c^i_r=\state_r^{\mem\Cc}(y_0x_0a_1\cdots a_i)$ be the state of
  $\mem\Cc_r$ reached on the prefix of $u$ up to $a_i$. Let
  $x'_i=\rep^{-1}(c^i_r,x_i)$. We get that the sequence
  \begin{equation}
    u'=y_0x'_0a_1y_1x'_1a_2\dots
  \end{equation}
  is a labeling of a maximal run in $\Cc$. The projections on processes other than $r$ are
  the same for $u$ and $u'$. It remains to see if the parity condition
  on $r$ is satisfied. We have $c_r^i \act{x_i} c_r^{i+1}$ in
  $\mem\Cc_r$ and $c_r^i \act{x'_i} c_r^{i+1}$ in $\Cc_r$. Since we
  lifted priorities to $\Cc$ and $\mem\Cc$ (being both covering),  the
  $r$-awareness of $\Aa$  lifts to $\Cc$ and $\mem\Cc$, so the
  same maximal rank is seen when reading $x_i$ and $x'_i$. This shows
  that the parity condition on $r$ is satisfied on the run of
  $\mem\Cc$ on $u$, since it is satisfied by the run of $\Cc$ on $u'$.

  Consider now a maximal run with finitely many communications between
  $q$ and $r$. There is an equivalent one labeled by a sequence of the
  form: 
  \begin{equation}
    u = y_0x_0a_1y_1x_1a_2\cdots a_ky_kx_k
  \end{equation}
  where $y_k$ and $x_k$ are potentially infinite. Since we have only
  modified the $r$-component of the controller, it must be $x_k$ that
  does not satisfy the parity condition on $r$. 

  Suppose first that $x_k=b_1b_2\cdots$ is infinite. Take the run
  $c_1\act{b_1} c_2\act{b_2} c_3\dots$ in $\mem\Cc_r$, where $c_1=\state_r^{\mem\Cc}(y_0x_0a_1\cdots a_k)$. We have a run
  $c_1\act{b_1} c'_2\act{x_2} c_2\act{b_2} c'_3\act{x_3} c_3 \cdots$ in
  $\Cc_r$, where $\rep(c'_i)=c_i$ and $x_i \in (\Sloc_r)^*$ is the accessibility path,
  as given by Remark~\ref{rem:path}. We have $c_i\sgt c'_{i+1}$ for all
  $i=1,2,\dots$ because there is an edge from $c_i$ to $c'_{i+1}$ in
  $\Ccloc_r$. Recall that $c_i=\rep(c'_i)$. The definition of
  representatives implies that either $\sig(c'_i)\geq \sig(c_i)$ or
  $c_i$ is in a strictly lower SCC than $c'_i$. Since lowering a
  component can happen only finitely many times we have
  $\sig(c'_i)\geq \sig(c_i)$ for all $i$ bigger than some $n$. We get
  $c_i\sgt c_{i+1}$ for $i>n$ which implies that $x_k$ satisfies the
  parity condition. A contradiction. 

If $x_k$ is finite then we define the sequence $u'=y_0x'_0a_1y_1x'_1a_2\cdots
a_ky_kx'_k$ in $L(\Cc)$, as in the first case. Since $u$ was maximal in $\mem\Cc$,
we have that $u'$ is maximal in $\Cc$ (if $r$ can do an action in
$u'$, the same can be done in $u$, since $u$ and $u'$ end in the same
state). Thus the $r$-state reached in $\Aa$ by $u$  belongs to $T_r$,
since this holds already for $u'$. We get again a contradiction.
\medskip

We will use Lemma~\ref{lemma-r-memoryless} to reduce the control
problem to that for $r$-short automata.

Given $\Aa$ we define a $r$-short automaton $\short\Aa$. All its
components will be the same but for the component $r$. The states $\short S_r$ of
$r$ will be sequences $w\in S^+_r$ of states of $\Aa_r$ without
repetitions, plus two new states $\top,\bot$. For a local transition
$s'_r\act{b} s''_r$ in $\Aa_r$ we have in $\short\Aa_r$ 
transitions:
\begin{align*}
  ws'_r\act{b}& ws'_rs''_r &\text{if $ws'_rs''_r$ a sequence without
    repetitions}\\
    ws'_r\act{b}& \top &\text{if $s''_r$ appears in $w$ and the
      resulting loop is even}\\
    ws'_r\act{b}& \bot &\text{if $s''_r$ appears in $w$ and the
      resulting loop is odd}
\end{align*}
There are also
communication transitions between $q$ and $r$:
\begin{equation*}
  (s_q,ws'_r)\act{b}(s'_q,s''_r)\qquad\text{if
    $(s_q,s'_r)\act{b}(s'_q,s''_r)$ in $\Aa$}
\end{equation*}
Notice that $w$ disappears in communication transitions. 
The  parity condition for $\short\Aa$ is also rather straightforward: it
is the same for the components other than $r$, and for $\Aa_r$ it is
\begin{itemize}
\item $\short\W(ws_r)=\W(s_r)$,
\item  $\short T_r=\set{\top}\cup\set{ws_r : s_r\in T_r}$.
\end{itemize}


\medskip

\textbf{Proof of Theorem~\ref{thm:short}:}
  Consider the implication from left to right.
  Let $\Cc$ be a correct covering controller for $\Aa$. 
  By Lemma~\ref{lemma-r-memoryless} we can assume that it is $r$-memoryless. 
  We show that $\Cc$ is also a covering correct controller for
  $\short\Aa$. We will concentrate on correctness, since the covering part
  follows by examination of the definitions. 

  Let us take some maximal run $\short\run(u)$ of $\short\Aa\times
  \Cc$, and suppose by contradiction that it does not satisfy the
  parity condition of $\short\Aa$.  By definition $\run(u)$ is a run
  of $\Aa\times\Cc$, but it may not be maximal. We have by
  construction of $\short\Aa$ that $\short\state_p(u)=\state_p(u)$ for
  $p\not=r$ and that $\state_r(u)$ is the last element of
  $\short\state_r(u)$. (Recall that $\short\state_p(u),\state_p(u)$
  denote the state reached on $u$ by 
  $\short\Aa\times \Cc$ and $\Aa\times \Cc$, resp.)

  Suppose that $\run(u)$ is not a maximal run of $\Aa\times \Cc$. We will
  extend it to a maximal run $\run(\bar u$). If $\run(u)$ ended in $\bot$ in the
  $r$-component of $\short \Aa$ then we could extend $\run(u)$ to a run of
  $\Aa\times\Cc$ not satisfying the parity condition (here we use that
  $\Cc$ is memoryless, so the odd loop in $\Aa$ exists also into one in $\Aa
  \times \Cc$). So the only
  other possibility is that $\run(u)$ ends in $\top$.  In this case it is
  possible to extend $\run(u)$ to a complete run of $\Aa\times\Cc$ by adding
  the even loop in the $r$-component. This makes $r$
  satisfy the parity condition. Let $\run(\bar u)$ be the resulting run.

  Now observe that if a parity condition for some process $p\not=r$ is
  violated on $u$ then on $\bar u$ the same condition is violated. If
  it is violated on $r$ then the only remaining possibility is that
  there are finitely many $r$-actions in $u$, and the state reached on
  $u$ is $ws_r$ with $s_r\not\in T_r$. But then $u$ is a maximal run
  of $\Aa \times \Cc$ and is not well terminated on $r$ either, a
  contradiction.

  For implication from right to left we take a covering controller $\short \Cc$
  for $\short\Aa$ and construct a controller $\Cc$ for $\Aa$. The
  controller $\Cc$ will be obtained by modifying the $r$-component of
  $\short \Cc$.  The states of $\Cc_r$ will be sequences of states
  of $\short\Cc_r$. They will be of bounded length. We will have that
  if $\short c_1\cdots \short c_k$ is a state of $\Cc_r$ then
  $\short\pi(\short c_k)$ is the state $s_1\cdots s_k$ of $\short\Aa_r$,
  where  $\short\pi(\short c_j)=s_1\cdots
  s_j$, for $j=1,\dots,k$. Moreover, we define $\pi(\short c_1 \cdots
  \short c_k)=s_k$. The transitions of $\Cc_r$ are
  \begin{itemize}
  \item $w\short c \act{b} w\short c \short d $\quad if $\short
    c \act{b} \short d$ in $\short\Cc_r$ and $\short \pi(\short d)\not=\top$. 
  \item $\short c_1\cdots \short c_k \act{b} \short c_1\cdots \short
    c_j$\quad if $\short c_k\act{b}\short c$ in $\short\Cc_r$, $\short
    \pi(\short c)=\top$ and $j$ is such that $\short\pi(\short c_k)$ is
    $s_1\cdots s_k$ with  $s_k\act{b} s_j$ in $\Aa$. 
  \end{itemize}
  Notice that since $\short\Cc$ satisfies the parity condition $\bot$
  cannot be reached.

\medskip

\begin{remark}
The construction of $\Cc$ guarantees that
  every sequence of local $r$-actions $x$ of $\Cc$ has a corresponding
  (possibly shorter) sequence $x'$ of $\short \Cc$. If the sequence in
  $\short\Cc$ starts in $s_1$ and finishes in $s_2$ then the sequence
  in $\Cc$ starts also in $s_1$, but now considered as a sequence of
  length $1$, and finishes in a sequence ending in $s_2$. Since $\Aa$
  is $r$-aware and $\Cc$ and $\short \Cc$ are both covering, this
  means that the maximal rank seen on both sequences is the same.
  
\end{remark}

  We need to show that all maximal runs of $\Aa\times\Cc$ satisfy the
  parity condition. For contradiction suppose that $\run(u)$ does not.

  If there are infinitely many communications between $q$ and $r$ on
  $u$ then we write it as
 \begin{equation}
    u=y_0x_0a_1y_1x_1a_2\dots
  \end{equation}
  where $a_i\in \S_q \cap \S_r$, $x_i\in (\Sloc_r)^*$, and $y_i\in (\S
  \setminus \S_r)^*$. Now for every $x_i$, Observation 1 gives $x'_i$
  so that the maximal ranks on $x_i$ and $x'_i$ are the same, so for
  \begin{equation}
    u'=y_0x'_0a_1y_1x'_1a_2\dots
  \end{equation} 
  $\run(u')$ is a maximal run of $\short\Cc$. This gives a run
  violating the parity condition of $\short \Cc$.

  If there are finitely many communications between $q$ and $r$ on
  $u$ then we write it as
  \begin{equation}
    u=y_0x_0a_1y_1x_1a_2\dots a_ky_kx_k
  \end{equation} where $y_k$ and $x_k$ are potentially infinite. The
  only complicated case is when $x_k$ is infinite. We need to show that
  the run of $\Cc_r$ on $x_k$ satisfies the parity condition. Recall
  that the states of $\Cc_r$ are  sequences of states of
  $\short\Cc_r$. Moreover the length of this 
  sequences is bounded. Take the shortest sequence appearing infinitely
  often in $x_k$. The biggest rank seen between consecutive appearances of this
  sequence is even, since the path can be decomposed into
  (several) even loops of $\Aa$.

\medskip

\begin{corollary}
  In the $r$-short plant $\short\Aa$ the $r$-controller may be chosen
  memoryless since there are no infinite local $r$-plays.\igw{remove corollary?}
\end{corollary}

\begin{remark}
  We claim that the complexity of the reduction from $\Aa$ to $\red\Aa$ is
  polynomial in the size of $\Aa_q$ and simply exponential in the size
  of $\Aa_r$. The reason is as follows. States of $\short\Aa_r$ are
  simple paths  (i.e., without repetition of states) of $\Aa_r$. When going from
  $\Aa$ to $\red\Aa$, the states of $\red\Aa_q$ contain $r$-local
  strategies $f :(\Sloc_r)^* \to \Ssys_r$. Putting things together, in $f$
  we deal with paths of $\short\Aa_r$ (mapping them to $\Ssys_r$). But the latter
  can be written more succinctly as paths of $\Aa_r$ without repetitions.
\end{remark}


\subsection{The reduced automaton $\red\Aa$}\label{sec:new} 
Equipped with the notions of covering controller and $r$-short strategy
we can now present the construction of the reduced automaton $\red
\Aa$. We suppose that $\Aa=\struct{\set{S_p}_{p\in
    \red\PP},s_{in},\set{\d_a}_{a\in\S}}$ is
$r$-short and we define now the reduced automaton $\red\Aa$ that results
by eliminating process $r$ (cf.\ Figure~\ref{fig:schema}).  Let
$\red\PP=\PP\setminus \set{r}$.  We construct
$\red\Aa=\struct{\set{\red{S_p}}_{p\in
    \red\PP},\red{s_{in}},\set{\red{\d_a}}_{a\in\red\S}}$ where the
components are defined below.

All the processes $p\not=q$ of $\red\Aa$ will be the same as in
$\Aa$. This means: $\red{S_p}=S_p$, and $\red\S_p=\S_p$. Moreover, all
transitions $\d_a$ with $\loc(a) \cap \{q,r\}=\es$ are as in
$\Aa$. Finally, in $\red\Aa$ the  correctness condition of $p\not=q$
is the same as in $\Aa$.

Before defining process $q$ in $\red\Aa$ let us introduce the notion
of $r$-local strategy. An \emph{$r$-local strategy from a state
  $s_r\in S_r$} is a partial function $f:(\Sloc_r)^*\to\Ssys_r$
mapping sequences from $\Sloc_r$ to actions from $\Ssys_r$, such that
if $f(v)=a$ then $s_r \act{va}$ in $\Aa_r$.  Observe that since the
automaton $\Aa$ is $r$-short, the domain of $f$ is finite.

Given an $r$-local strategy $f$ from $s_r$, a local action $a \in
\Sloc_r$ is \emph{allowed by $f$}
if $f(\epsilon)=a$, or $a$ is uncontrollable. For $a$ allowed by $f$
 we
denote by $f_{|a}$ the $r$-local strategy defined by
$f_{|a}(v)=f(av)$; this is a strategy from $s'_r$, where $s_r \act{a}
s'_r$.


The states of process $q$ in $\red\Aa$ are of one of the following types:
\begin{equation*}
\struct{s_q,s_r}\,,\quad \struct{s_q,s_r,f}\,,\quad  \struct{s_q,a,s_r,f}  \,,
\end{equation*}
where $s_q \in S_q, s_r\in S_r$, $f$ is a $r$-local strategy from $s_r$,
and $a\in\Sloc_q$. The new initial state for $q$ is 
$\struct{(s_{in})_q,(s_{in})_r}$. 
Recall that that since $\Aa$ is $r$-short, any $r$-local strategy in $\Aa_r$ is
necessarily finite, so $\red S_q$ is a finite set. Recall also
  that controllable actions are local.

The transitions of $\red\Aa_q$ are presented in
Figure~\ref{fig:transitions}. Transition \circled{$1$} chooses an
$r$-local strategy $f$. It is followed by transition \circled{$2$}
that declares a controllable action $a\in\Ssys_q$ that is enabled from
$s_q$. Transition \circled{$3$} executes the chosen action $a$; we
require $s_q\act{a} s'_q$ in $\Aa_q$. Transition \circled{$4$}
executes an uncontrollable local action $b_4\in\Senv_q$; provided
$s_q\act{b_4} s''_q$ in $\Aa_q$. Transition \circled{$5$} executes a local action
$b_5\in\Sloc_r$, provided that $b_5$ is allowed by $f$ and $s_r\act{b_5} s'_r$.
Transition \circled{$6$} simulates a synchronization $b_6$ between $q$ and
$r$; provided $(s_q,s_r)\act{b_6}(s'_q,s'_r)$ in $\Aa$. Finally, transition \circled{$7$} simulates a synchronization
between $q$ and $p\not=r$. An example of a simulation of $\Aa_q$ and
$\Aa_r$ by $\red\Aa_q$ is presented in
Figure~\ref{fig:simulation}. The numbers below transitions refer to
the corresponding cases from the definition. 








To summarize, in $\red\S_q$ we have all actions of $\S_r$ and $\S_q$,
but they become uncontrollable. All the new actions of
process $q$ in plant $\red\Aa$ are \emph{controllable}:
\begin{itemize}
\item action $\ch(f) \in \Ssys$, for every local $r$-strategy $f$,
\item action $\ch(a)$, for every $a\in\Ssys_q$.
\end{itemize}
\begin{figure}[tbp]
  \centering
\includegraphics[scale=1]{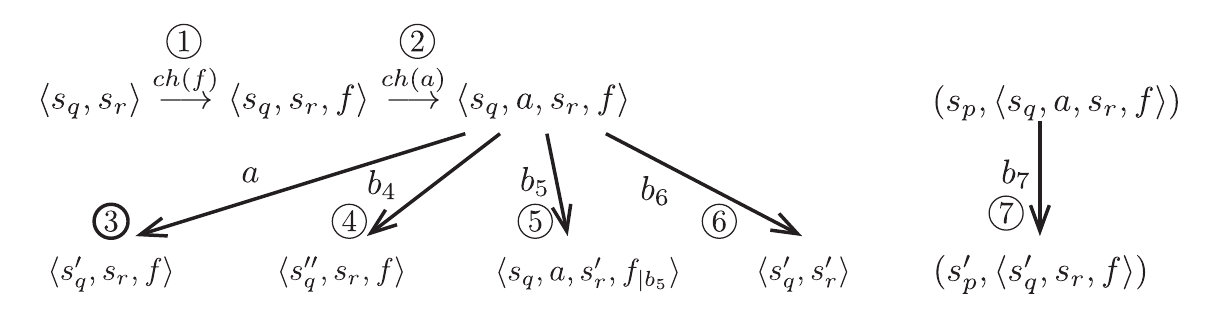}
  \caption{Transitions of $\red\Aa$}\label{fig:transitions}
\end{figure}

\begin{figure}[tbp]
  \centering
\includegraphics[scale=.6]{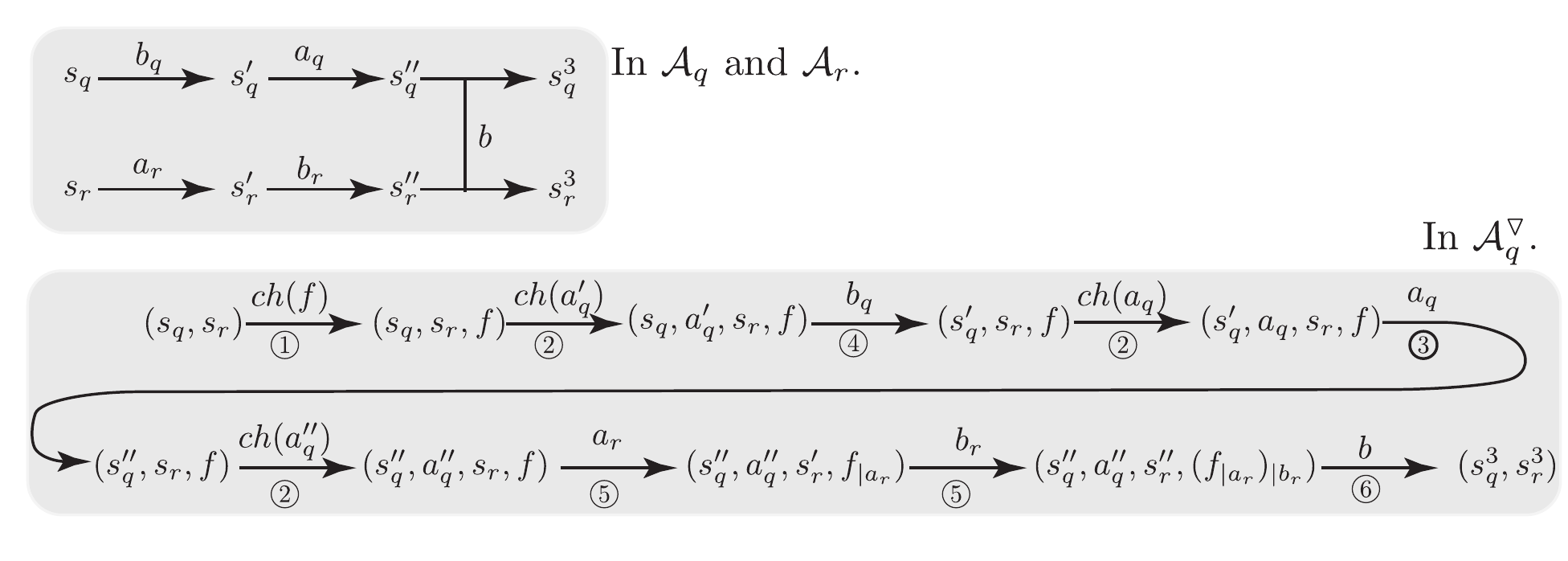}
  \caption{Simulation of $\Aa_q$ and $\Aa_r$ by $\red\Aa_q$.}\label{fig:simulation}
\end{figure}
The correctness condition for process $q$ in $\red\Aa$ is:
\begin{enumerate}
\item The correct infinite runs of $q$ in $\red\Aa$ are those that have
  the projection on transitions of $\Aa_q$ correct with respect to $\Cor_q$, and
  either: \emph{(i)} the projection on transitions of $\Aa_r$ is infinite and
  correct with respect to $\Cor_r$; or \emph{(ii)} the projection on
  transitions of $\Aa_r$ is finite and for $f,s_r$ appearing in almost
  all states of $q$ of the run we have that from $s_r$ all sequences
  respecting strategy $f$ end in a state from $T_r$.
\item $\red T_q$ contains states $\struct{s_q,s_r,f}$ such that
  $s_q\in T_q$, and $s_r\in T_r$.
\end{enumerate}
Item $1 (ii)$ in the  definition above captures 
the case where $q$ progresses alone till infinity and blocks
$r$, even though $r$ could reach a terminal state in a couple of
moves. Clearly, item $1$ can be expressed as an $\omega$-regular
condition. \igw{added some explanations} The definition of correctness condition is one of the
principal places where the $r$-short assumption is used. Without this
assumption we would need to cope with the situation where we have an
infinite execution of $\Aa_q$, and at the same time an infinite
execution of $\Aa_r$ that do not communicate with each other. In this
case $\red\Aa_q$ would need to fairly simulate both executions in some way.


The reduction is rather delicate since in concurrent systems there are
many different interactions that can happen.  For example, we need to
schedule actions of process $q$, using $\ch(a)$ actions, before the
actions of process $r$. The reason is the following. First, we need to
make all $r$-actions uncontrollable, so that the environment could
choose any play respecting the chosen $r$-local strategy. Now, if we
allowed controllable $q$-actions to be eligible at the same time as
$r$-actions, then the control strategy for automaton $\red\Aa$ would
be to propose nothing and force the environment to play the
$r$-actions. \igw{added some explanations} This would allow the
controller of $\red\Aa$ to force the advancement of the simulation of
$r$ and get information that is impossible to obtain by the controller
of $\Aa$.

Together with Theorem~\ref{thm:short}, the theorem below implies 
our main
Theorem~\ref{thm:main}. 

\begin{theorem}\label{thm:correctness}
  For every $r$-short Zielonka automaton $\Aa$ and every local,
  $\omega$-regular
  correctness conditions: there is a correct covering controller for
  $\Aa$ iff there is a correct covering controller for $\red\Aa$. 
  The size of $\red\Aa_q$ is polynomial in the size of $\Aa_q$ and
  exponential in the size of $\Aa_r$.
\end{theorem}

We end with some notations used in the following sections. For a
Zielonka automaton $\red\Cc$ over $(\red\S,loc)$ and $w
\in(\red\S)^\infty$ we will write $\red\run(w)$ for the sequence of transitions of $\red\Cc$ when
  reading $w$. For finite $w$ we will write $\state(w)$ for the last state in
  $\red\run(w)$.

\subsection{Proof of Theorem~\ref{thm:correctness}: from $\Cc$
                                to $\red\Cc$}\label{sec:C}

By Lemma \ref{lemma:covering} we can assume that we have a
\emph{correct covering controller}  $\Cc$ for $\Aa$. We show how to construct
a correct controller $\red\Cc$ for $\red\Aa$. This will give the left
to right implication of Theorem~\ref{thm:correctness}.

\begin{remark} Some simple observations about $\Cc$.
\begin{enumerate}
\item We may assume that from every state of $C$ there is at most one
  transition on a local controllable action. If there were more than
  one, we could arbitrary remove one of them. This will reduce
  the number of maximal runs so the resulting controller with stay correct.

\item $\Cc$ determines for every state $c$ of $\Cc_r$ a local
  $r$-strategy $f$ from $\pi(c)$: if $c=c_0 \act{a_1} c_1 \act{a_2}
  \cdots \act{a_k} c_k$, $\pi(c_i)=s_i$ and $a_i \in \Sloc_r$ for all
  $i$, then $f(c_0 \cdots c_k)=a$ where $a\in\Ssys_r$ is a (unique)
  controllable action possible from $c_k$. This strategy
  may have memory, but all the (local) plays respecting
  $f$ are of bounded length, assuming that $\Aa$ is $r$-short.
\end{enumerate}
\end{remark}

The components $\red\Cc_p$ for $p\not=q$ are just $\Cc_p$, and the
initial state is the same. The
component $\red\Cc_q$ is described below. Its states are of the form
$(c_q,c_r)$, $(c_q,c_r,f)$ and $(c_q,a,c_r,f)$ with $c_q \in C_q$,
$c_r \in C_r$, $a \in \Ssys_q$, and local $r$-strategy $f$. Its
initial state is $(c^0_q,c^0_r)$, with $c^0_q,c^0_r$ initial
states of $\Cc_q,\Cc_r$. 

The transitions of $\red\Cc_q$ ensure the right choice of a local
strategy and of a local action:

\begin{itemize}
\item Choice of $r$-strategy:
  \[ (c_q,c_r)\act{\ch(f)}(c_q,c_r,f)\] where $f$ is the local
  $r$-strategy from $\pi(c_r)$ determined by $\Cc$ in state $c_r$.

\item Choice of a (local) controllable $q$-action:
  \[(c_q,c_r,f)\act{\ch(a)} (c_q,a,c_r,f)\] For $a \in \Ssys_q$ unique such that
  $c_q\act{a}c_q'$, for some $c'_q$. If there is no such transition then we
  put some arbitrary fixed action $a_0\in \Ssys_q$.
\end{itemize}

The other transitions of $\red\Cc_q$ are on uncontrollable actions,
they just reflect the structure of $\red\Aa$:
\begin{itemize}
\item Execution of the chosen controllable $q$-action:
\[ (c_q,a,c_r,f)\act{a}(c'_q,c_r,f)\qquad
\text{if $c_q\act{a}c_q'$ in $\Cc_q$.} \]

\item Execution of an uncontrollable local $q$-action:
  \[ (c_q,a,c_r,f)\act{b}(c'_q,c_r,f)\qquad \text{if $c_q\act{b}c_q'$
    in $\Cc_q$, where $b \in \Senv_q \cap \Sloc_q$.} \]

\item Communication between $q$ and $p\not=r$:
\[(c_p,(c_q,a,c_r,f))\act{b}(c_p',(c_q',c_r,f))\quad
\text{if $(c_p,c_q)\act{b} (c_p',c_q')$ in $\Cc$.}\]

\item Local move of $r$:
\[(c_q,a,c_r,f)\act{b} (c_q,a,c_r',f)\quad \text{if $c_r\act{b}c_r'$ in
  $\Cc_r$, where $b \in \Sloc_r$.} \]

\item Communication between $q$ and $r$
\[(c_q,a,c_r,f)\act{b}(c_q',c_r')\quad
\text{if $(c_q,c_r)\act{b}(c_q',c_r')$ in $\Cc$.}\]
\end{itemize}

\begin{lemma}\label{lemma:C is covering}
  If $C$ is a covering controller for $\Aa$ then $\red\Cc$  is a
  covering controller for $\red\Aa$. The covering function is
  \begin{align*}
    \red\pi(c_q,c_r)=& (\pi(c_q),\pi(c_r))\\
    \red\pi(c_q,c_r,f)=& (\pi(c_q),\pi(c_r),f)\\
    \red\pi(c_q,a,c_r,f)=& (\pi(c_q),a,\pi(c_r),f)\ .
  \end{align*}
\end{lemma}

For the correctness proof we will need one more definition:

\begin{definition}[$\hide$]
  For $w\in(\red\S)^\infty$ we let $\hide(w)\in\S^\infty$ be
  the sequence obtained by removing actions from $\red\S
  \setminus\S$.
\end{definition}

Observe that by construction of $\red\Cc$ if $\red\run(w)$ is defined
then in $w$ there can be at most two consecutive $q$-actions from
$\red\S \setminus \S$.

\begin{lemma}\label{lemma:C invariant}
  Let $w\in (\red\S)^*$. If $\red\run(w)$ is defined then so is
  $\run(\hide(w))$.  Moreover, letting $\red c=\red\state(w)$ and
  $c=\state(\hide(w))$, we have that
 (i) $\red c_p=c_p$  for all $p\not=q,r$, and (ii) 
  $\red c_q$ is either $(c_q,c_r)$, or $(c_q,c_r,f)$, or
  $(c_q,a,c_r,f)$; where $a$ and $f$ are 
  determined by $c_q$ and $c_r$ as follows:
  \begin{itemize}
  \item $a$ is the unique controllable $q$-action from $c_q$ in $\Cc$ (or
    $a_0$ if there is none).
  \item $f$ is the local $r$-strategy determined by $\Cc$
    in $c_r$.
  \end{itemize}
\end{lemma}
\begin{proof}
  The proof is by induction on the length of $w$. It follows by direct
  examination of the rules.
\end{proof}

\begin{lemma}\label{lemma:projecting-runs}
  Assume that $w\in( \red \S)^\infty$. For every process $p\not=q$ we have
  $\red\run_p(w)=\run_p(\hide(w))$. Concerning
  $\red\run_q(w)$: if we project it on transitions of $\Cc_q$ we obtain
  $\run_q(\hide(w))$; if we project it  on transitions of
  $\Cc_r$ we obtain $\run_r(\hide(w))$.
\end{lemma}
\begin{proof}
  Directly from the previous lemma.
\end{proof}


\begin{lemma}\label{lemma:C correct}
  If $\Cc$ is a correct covering controller for $\Aa$ then $\red\Cc$
  is a correct covering controller for $\red\Aa$.
\end{lemma}
\begin{proof}
  Since $\Cc$ is a correct covering controller we have that
  all maximal runs of $\Cc$ are correct w.r.t~$\Aa$.  By
  Lemma~\ref{lemma:C is covering} we 
  know that $\red \Cc$ is a covering controller,
  so it is enough to show that all maximal runs of $\red\Cc$ are
  correct w.r.t.~$\red\Aa$.

  Take a maximal run in $\red\Cc$, say on $w\in (\red\S)^\infty$. The
  first obstacle is that $\run(\hide(w))$ may be not maximal in
  $\Cc$. This can only happen when there are infinitely many
  $q$-actions in $w$, but only finitely many $r$-actions. Then we have
  $w=v_1v_2$ and there are no $r$-actions in $v_2$. Let
  $\red\state_q(v_1)=(c_q,a,c_r,f)$. We have that $c_r$ and $f$ appear
  in all $\red\state_q(v_1v')$, for every prefix $v'$ of $v_2$. The
  run $\run(\hide(w))$ is not maximal when there is at least some
  local action of $\Cc_r$ enabled in $c_r$. Let $x$ be a maximal
  sequence of local $r$-actions that is possible in $\Cc_r$ from state
  $c_r$. Since $\Aa$ is $r$-short, every such sequence is
  finite. Moreover we choose $x$ in such a way that it brings $\Cc_r$
  into a state not in $T_r$ (if it is possible).  We get that
  $u=v_1xv_2$ also defines a maximal run of $\red\Cc$, but now the run
  on $\hide(u)$ is maximal in $\Cc$. Notice that $\red\run(u)$
  satisfies $\red\Cor$ iff $\red\run(w)$ does: the difference is the
  sequence $x$, and we have chosen, if possible, a losing sequence.

  We need to show that the run of $\red\Aa$ on $u$ satisfies
  $\red\Cor$ using the fact that the run on $\hide(u)$ satisfies
  $\Cor$.  For $p\not=q$, Lemma~\ref{lemma:projecting-runs} tells us
  that $\red\run_p(u)$ is the same as $\run_p(\hide(u))$. Since
  $\red\Cor_p$ and $\Cor_p$ are the same, we are done.

  It remains to consider $\red\run_q(u)$. If there are finitely many
  $q$-actions  in $u \in (\red\S)^\infty$ then  $u=u_1u_2$ with no
  $q$-action in $u_2$. Consider
  $\red\state_q(u_1)=(c_q,a,c_r,f)$. We have that
  $\state_q(\hide(u_1))=c_q$ and $\state_r(\hide(u_1))=c_r$. As there
  are no $q$-actions in $u_2$, and $\run_q(u)$ satisfies
  $\Cor_q$, we must have $\pi(c_q)\in T_q$ and $\pi(c_r)\in T_r$. This shows that
  $\red\run_q(u)$ satisfies $\red\Cor_q$.

  If there are infinitely many $q$-actions in $u\in(\red\S)^\infty$, we
  still have two cases. The first is when there are infinitely many
  actions from $\S_r$ as well. Then $\red\run_q(u)$
  satisfies $\red\Cor_q$ if the corresponding runs $\run_q(\hide(u))$
  and $\run_r(\hide(u))$ satisfy $\Cor_q$ and $\Cor_r$,
  respectively. This is guaranteed by our assumption that $\run(\hide(u))$
  satisfies $\Cor$.

  The last case is when in $u\in(\red\S)^\infty$ we have infinitely
  many $q$-actions and only finitely many actions from $\S_r$. Then $u=u_1u_2$
  with no actions from $\S_r$ in $u_2$. We get $\red\state_q(u_1)=(c_q,a,c_r,f)$
  with both $c_r$, $f$ appearing in all the further states of the
  run. Since $\run_r(\hide(u))$ satisfies $\Cor_r$, we have that
  $\pi(c_r) \in T_r$. But then, by the construction of $u$, there is
  no $\S_r$-transition possible from $c_r$ (and neither from
  $\pi(c_r)$ in $\red\Aa_r$, since $\red\Cc$ is covering). This means that $\red\run_q(u)$
  satisfies $\red\Cor_q$.

\end{proof}

\subsection{Proof of Theorem~\ref{thm:correctness}: from $\red \Dd$
                                to $\Dd$}\label{sec:D}
This subsection gives the
right-to-left direction of the proof. Given a correct controller $\red\Dd$ for
$\red\Aa$, we show how to construct a correct controller $\Dd$ for
$\Aa$. By Lemma~\ref{lemma:covering} we can assume that $\red\Dd$ is
covering.


The components $\Dd_p$ for $p\not=q,r$ are the same as in
$\red\Dd$. So it remains to define
$\Dd_q$ and $\Dd_r$. 
The states of $\Dd_q$ and $\Dd_r$ are obtained from states of
$\red\Dd_q$. We need only certain states of $\red\Dd_q$,
namely those $d_q$ whose projection $\red\pi(d_q)$ in $\red\Aa_q$ has four
components, we call them \emph{true states} of $\red\Dd_q$:
\begin{equation*}
  \ts(\red\Dd_q)=\set{d_q\in \red\Dd_q \mid \red\pi(d_q)\ \text{is of the
      form $(s_q,a,s_r,f)$}}.
\end{equation*}
Figure~\ref{fig:d-simulation}
presents an execution of $\red\Aa$ controlled by $\red\Dd$. We can see
that $d_2$ is a true state, and $d_3$ is not.

The set of states of $\Dd_q$ is just $\ts(\red\Dd_q)$, while the states of
$\Dd_r$ are pairs $(d_q,x)$ where $d_q$ is a state from
$\ts(\red\Dd_q)$ and $x \in (\Sloc_r)^*$ is a sequence of local
$r$-actions  that is possible from $d_q$ in $\red\Dd$, in
symbols $d_q\act{x}$. We will argue later that such sequences are
uniformly bounded. The
initial state of $\Dd_q$ is the state $d^1_q$ reached from the initial state
of $\red\Dd_q$ by the (unique) transitions of the form
$\ch(f_0),\ch(a_0)$. The initial state of $\Dd_r$ is $(d^1_q,\e)$. 
The local transitions for $\Dd_r$ are 
  $(d_q,x)\act{b}(d_q,xb)$, for every $b\in\S^{loc}_r$
    and $d_q\act{xb}$.

Before defining the transitions of $\Dd_q$ let us observe that if 
$d_q\in\red\Dd_q$ is not in $\ts(\red\Dd_q)$ then only one
controllable transition is possible from it. Indeed, as $\red\Dd$ is a
covering controller, if $\red\pi(d_q)$ is of the form $(s_q,s_r)$ then
there can be only an outgoing transition on a letter of the form
$\ch(f)$. Similarly, if $\red\pi(d_q)$ is of the form $(s_q,s_r,f)$
then only a $\ch(a)$ transition is possible. Since both $\ch(f)$ and
$\ch(a)$ are controllable, we can assume that in $\red\Dd_q$ there is
no state with two outgoing transitions on a letter of this form. For a
state $d_q\in\red\Dd_q$ not in $\ts(\red\Dd_q)$ we will denote by
$\ts(d_q)$ the unique state of $\ts(\red\Dd_q)$ reachable from $d_q$
by one or two transitions of the kind $\act{\ch(f)}$ or $\act{\ch(a)}$,
depending on the cases discussed above. Going back to
Figure~\ref{fig:d-simulation}, we have $\ts(d_3)=d_4$. 

We now describe the $q$-actions possible in $\Dd$.
\begin{itemize}
\item Local $q$-action $b \in\Sloc_q$: $d_q\act{b} \ts(d'_q)$ if
  $d_q\act{b} d'_q$ in $\red\Dd_q$. For example, this gives a
  transition $d_1\act{b_q} d_4$ in Figure~\ref{fig:d-simulation}.

\item Communication $b \in\S_q \cap \S_p$ between $q$ and $p\not=r$: $(d_p,d_q)\act{b}
  (d'_p,\ts(d'_q))$ if $(d_p,d_q)\act{b}(d'_p,d'_q)$ in $\red\Dd$.

\item Communication $b \in\S_q \cap \S_r$ of $q$ and $r$:
  $(d^1_q,(d^2_q,x))\act{b}(\ts(d''_q),(\ts(d''_q),\e))$ 
if $d^1_q \act{x} d'_q \act{b} d''_q$ in $\red\Dd_q$;
observe that $\act{x}$ is a sequence of transitions. For example, this gives a
  transition on $b$ in Figure~\ref{fig:d-simulation}.
\end{itemize}
In the last item the transition does not depend on $d^2_q$ since,
informally, $d^1_q$ has been reached from $d^2_q$ by a sequence of
actions independent of $r$. The condition $d^1_q \act{x} d'_q \act{b}
d''_q$ simulates the order of actions where all local $r$-actions come
after the other actions of $q$, then we add a communication between
$q$ and $r$.

\begin{figure}[tbp]
\includegraphics[scale=.7]{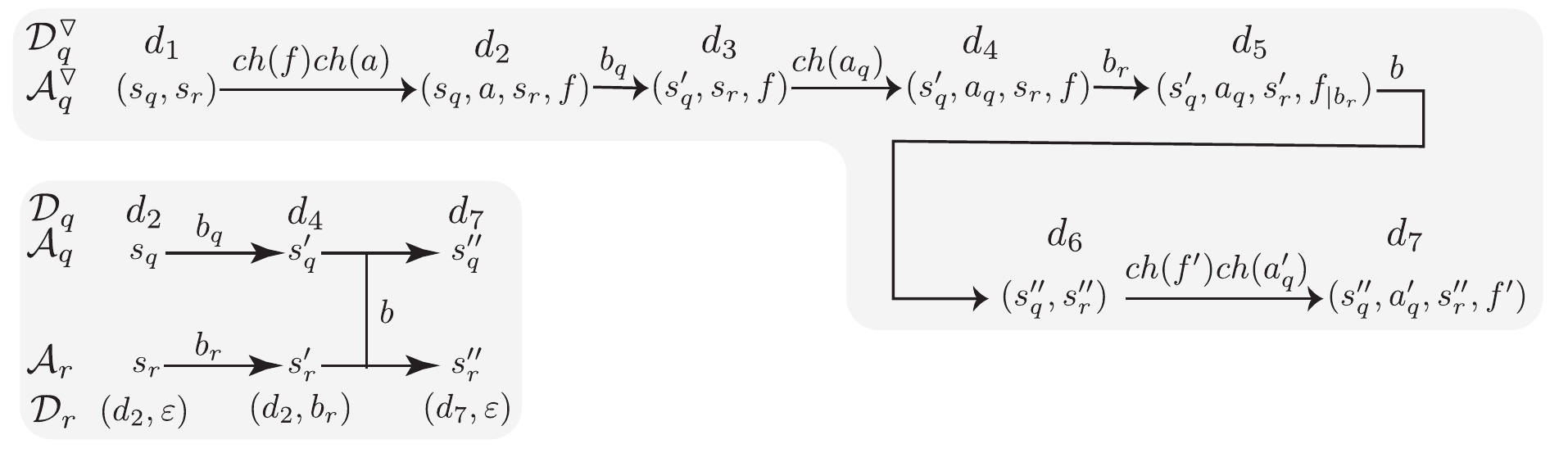}
  \caption{Decomposing controller $\red\Dd_q$ into $\Dd_q$ and
    $\Dd_r$.}\label{fig:d-simulation}
\end{figure}

The next lemma says that $\Dd$ is a covering controller for
$\Aa$. Since $\Aa$ is assumed to be $r$-short, the lemma also gives a
bound on the length of sequences in the states of $\Dd_r$. 
\begin{lemma}\label{lem:Dcovering}
  If $\red\Dd$ is a covering controller for $\red \Aa$ then $\Dd$ is a
  covering controller for $\Aa$.
\end{lemma}

\begin{proof}
  We need to define the projection function $\pi$ using the projection
  function $\red\pi$. For $p\not=q,r$ set $\pi=\red\pi$. For $\Dd_q$
  we define $\pi(d_q)=s_q$ where $s_q$ is the state of $\Aa_q$ in
  $\red\pi(d_q)$.  For $\Dd_r$ and its state $(d_q,x)$ we define
  $\pi(d_q,x)=s'_r$ where $d_q\act{x}d'_q$ and $s'_r $ is the state of $\Aa_r$ in
  $\red\pi(d'_q)$.

  We need to check that the transitions defined above preserve this
  projection function; namely for every process $p$: if $d_p\act{b}
  d'_p$ in $\Dd_p$ then $\pi(d_p)\act{b}\pi(d'_p)$ in $\Aa_p$; and
  similarly for communication actions.  The statement is obvious if
  the move is in components other than $q$ or $r$. We are left with
  four cases:
  \begin{itemize}
  \item Local move of $q$, namely $d_q\act{b}d'_q$. We have
    $d_q\act{b}d''_1\act{\ch(a')} d'_q$ in $\red\Dd$ for some $a'$, since
    $d'_q=\ts(d''_q)$.  By the fact that
    $\red\Dd$ covers $\red\Aa$ and the definition of moves of the
    latter automaton we
    have in $\red\Aa$:
    \begin{equation*}
      \red\pi(d_q)=\struct{s_q,a,s_r,f}\act{b}\struct{s'_q,s_r,f}\act{\ch(a')}
      \struct{s'_q,a',s_r,f}=\red\pi(d'_q) \,,
    \end{equation*}
    and by definition of $\red\Aa$ we know that $s_q\act{b}s'_q$ is in $\Aa$.
  \item Communication between $q$ and $p\not=r$ is similar.
  \item Local move of $r$: $(d_q,x)\act{b}(d_q,xb)$. By definition we
    know that from $d_q$ it is possible to do in $\red\Dd$ the sequence of actions
    $xb$, that is $d_q\act{x} d^1_q\act{b}d^2_q$. \anca{old: We have
    $\red\pi(d_q)=(s_q,a,s_r,f)$, $\red\pi(d^1_q)=(s_q,a,s^1_r,f)$
    and $\red\pi(d^2_q)=(s_q,a,s^2_r,f)$;} 
    We have
    $\red\pi(d_q)=(s_q,a,s_r,f)$, $\red\pi(d^1_q)=(s_q,a,s^1_r,f|_x)$
    and $\red\pi(d^2_q)=(s_q,a,s^2_r,f|_{xb})$; since $xb$ is a sequence of
    local $r$-actions the other components do not change. We have
    $s^1_r\act{b}s^2_r$ by definition of $\red\Aa_q$, and $\red\pi(d_q,x)=s^1_r$,
    $\red\pi(d_q,xb)=s^2_r$, as required.  
  \item Communication between $q$ and $r$:
    $(d^1_q,(d^2_q,x))\act{b}(\ts(d''_q),(\ts(d''_q),\e))$. By
    definition this is possible only  when $d^1_q \act{x} d'_q \act{b} d''_q$ in
    $\red\Dd_q$. Since $\red\Dd$ is covering we get the following
    sequence of transitions in $\red\Aa$: \anca{changed f in f' below,
      f' in g}
    \begin{gather*}
      \red\pi(d^1_q)=\struct{s_q,a,s_r,f}\act{x}
      \struct{s_q,a,s^1_r,f|_x}\act{b}\struct{s'_q,s'_r}
      \act{\ch(g)}\\
      \struct{s'_q,s'_r,g}\act{\ch(a')}\struct{s'_q,a',s'_r,g}=
      \red\pi(\ts(d''_q)) 
    \end{gather*}
  \end{itemize}
  So we have $(s_q,s^1_r)\act{b}(s'_q,s'_r)$ in $\Aa$ and
  $\pi(d^1_q)=s_q$, $\pi(\ts(d''_q))=s'_q$, $\pi(\ts(d''_q),\e)=s'_r$.
  We claim that $\pi(d^2_q,x)=s^1_r$, and for this
   we need to observe a property of the runs of
  $\Dd$ (proved by induction on the length of the run). The intuition
  for the property below is that $d_q$ was reached from $d'_q$ by
  actions that do not involve $r$. \anca{added ``The intuition..''}

  \begin{quote}
    \textbf{Property (*)} If from the initial state $\Dd$ can reach a
  global state with $d_q$ and $(d'_q,x)$ at the coordinates
  corresponding to $q$ and $r$, respectively, then the $s_r$- and $f$-components of the $\red\pi$ projections of $d_q$ and $d'_q$ are the same:
 $\red\pi(d_q)=(s_q,a,s_r,f)$ and
  $\red\pi(d'_q)=(s'_q,a',s_r,f)$, for some $s_q,s'_q,a,a',s_r,f$.

  \end{quote}

  From Property (*) it follows that $\pi(d^2_q,\e)=s_r$, hence
  $\pi(d^2_q,x)=s^1_r$ since $s_r\act{x} s^1_r$.

  It remains to check the controllability condition for $\Dd$. For
  components other than $q$ and $r$ this is obvious. 
  We have four cases to examine.

  First, let us take a state $(d_q,x)$ of $\Dd_r$. Suppose that
  $\pi(d_q,x)\act{b}s'_r$ is a local, uncontrollable transition in
  $\Aa_r$. We need to show that $(d_q,x)\act{b}(d_q,xb)$ is possible
  in $\Dd_r$. Since $(d_q,x)$ is a state of $\Dd_r$ we have
  $d_q\act{x}d'_q$ in $\red\Dd_q$. Moreover, $\red\pi(d'_q)$ is of the
  form $(s_q,a,s_r,f)$ and $\pi(d_q,x)=s_r$. We get that
  $(s_q,a,s_r,f)\act{b} (s_q,s'_r,f|_b)$ \anca{changed 2nd f in f'}exists in $\red\Aa_q$.  Since
  $\red\Dd$ satisfies the controllability condition, in $\red\Dd_q$ there must
  be a transition $d'_q\act{b}d''_q$ for some $d''_q$. Hence, by
  definition, $(d_q,x)\act{b}(d_q,xb)$ exists in $\Dd_r$.

  For the next case we take a state $d_q$ of $\Dd_q$ and suppose that
  $\pi(d_q)\act{b} s'_q$ is a local, uncontrollable transition in
  $\Aa_q$. We need to show that a $b$-transition is possible from
  $d_q$ in $\Dd_q$. We get $\red\pi(d_q)$ is of the form $(s_q,a,s_r,f)$, and
  $\pi(d_q)=s_q$. This means that the transition
  $(s_q,a,s_r,f)\act{b}(s'_q,s_r,f)$ is in $\red\Aa_q$. Since
  $\red\Dd$ is covering, we get $d_q\act{b}d'_q$ for some $d'_q$ in
  $\red\Dd_q$. But then $d_q\act{b}\ts(d'_q)$ in $\Dd_q$ by
  definition.

  The case of communication of $q$ with $p\not=r$ is similar to the
  above.

  The last case is a communication between $q$ and $r$. So take
  $(d^1_q,(d^2_q,x))$ and suppose
  $(\pi(d^1_q),\pi(d^2_q,x))\act{b}(s'_q,s'_r)$ in $\Aa$. We have that
  $\red\pi(d^1_q)$ is of the form $(s^1_q,a_1,s_r,f)$ and
  $\red\pi(d^2_q)$ is of the form $(s^2_q,a_2,s_r,f)$; the $s_r$- and
  $f$-components are the same by Property (*). Moreover,
  by definition $\pi(d^1_q)=s^1_q$ holds. Let $s^1_r=\pi(d^2_q,x)$, thus
  $s_r\act{x}s^1_r$. These observations allow us to
  obtain the following sequence of transitions in $\red\Aa$:
  \anca{changed 2nd f in f'}
  \begin{equation*}
    (s^1_q,a_1,s_r,f)\act{x}(s^1_q,a_1,s^1_r,f|_x)\act{b}(s'_q,s'_r)
  \end{equation*}
  Since $\red\Dd$ satisfies the controllability condition   we must
  have transitions $d^1_q\act{x}d'_q\act{b}d''_q$ in $\red\Dd$, 
  with $\red\pi(d''_q)=(s'_q,s'_r)$. This means that we have transition
  $(d^1_q,(d^2_q,x))\act{b}(\ts(d''_q),(\ts(d''_q),\e))$ in $\Dd$ and 
  $\pi(\ts(d''_q))=s'_q$, $\pi(\ts(d''_q),\e)=s'_r$. 
\end{proof}

As $\Dd$ is covering, to prove that $\Dd$ is correct we need to show
that all its maximal runs satisfy the correctness condition. For this
we will construct for every run of $\Dd$ a corresponding run of $\red
\Dd$. The following definition and lemma tells us that it is enough to
look at the runs of $\Dd$ of a special form. 
\begin{definition}[$\slow$]\label{def:slow}
  We define $\slow_r(\Dd)$ as the set of all 
  sequences labeling runs of $\Dd$ of the form $ y_0x_0a_1\cdots a_k
  y_k x_k a_{k+1}\dots$ or $y_0x_0a_1\cdots y_{k-1}x_{k-1} a_k x_k
  y_\omega$, 
where $a_i\in \S_q \cap \S_r$, $x_i\in (\S^{loc}_r)^*$,  $y_i\in
(\S \setminus \S_r)^*$, and $y_\omega\in
(\S \setminus \S_r)^\omega$
\end{definition}

\begin{lemma}\label{lemma:reduciton-to-slow}
 A covering controller $\Dd$ is correct for $\Aa$ iff
  for all $w\in\slow_r(\Dd)$, $\run(w)$ satisfies the correctness condition
  inherited from $\Aa$.
\end{lemma}

\begin{proof}
  Observe first $\Dd$ is
  $r$-short, since $\Aa$ is $r$-short and $\Dd$ is covering. Thus
  every sequence labeling some run of $\Dd$  either has
  finitely many $r$-actions or infinitely many communications of $r$
  with $q$.  

  Secondly, note that every sequence $w$ labeling
  some run of $\Dd$ can be rewritten into a sequence $w'$ from
  $\slow_r(\Dd)$ by repeatedly replacing factors $ab$ by $ba$, if
  $\loc(a) \cap \loc(b)=\es$. We have that $\run(w')$ is also
  defined and $\run_p(w)=\run_p(w')$ for every
  process $p$. Therefore for
  correctness it will be enough to reason on sequences from
  $\slow_r(\Dd)$.
\end{proof}

For every sequence $w\in\slow_r(\Dd)$ as in Definition~\ref{def:slow} we
define the sequence $\chi(w) \in (\red\S)^\infty$ by induction on the
length of $w$. Let $ \chi(\e)=\ch(f_0)\ch(a_0)$, where $f_0$ and $a_0$
are determined by
    the initial $q$-state of $\red\Dd$.  For $w \in \S^*, b\in \S$ let
\begin{equation*}
  \chi(wb)=
\begin{cases}
\chi(w)b & \text{if $b\not\in\S_q$}\\
\chi(w)b\ch(a) & \text{if $b\in\S_q \setminus \S_r$}\\
\chi(w)b\ch(f)\ch(a) & \text{if $b\in\S_{q}\cap\S_r$.}
\end{cases}
\end{equation*}
where $a$ and $f$ are determined by the state reached by $\red\Dd$ on
$\chi(w)b$.   The next lemma
implies the correctness of the construction, and at the same time
confirms that the above definition makes sense, that is, the needed
runs of $\red\Dd$ are defined.



\begin{lemma}\label{lemma:D invariant}
  For every sequence $w\in \slow_r(\Dd)$ we have that
  $\red\run(\chi(w))$ is defined. If $w$ is finite then the states
  reached by $\Dd$ on $w$ and by $\red\Dd$ on $\chi(w)$ satisfy the following:
  \begin{enumerate}
  \item $\state_p(w)=\red\state_p(\chi(w))$ for every
  $p\not=q,r$.
\item Let $w=y_0x_0a_1 \cdots a_ky_kx_k$, where $a_i\in \S_q \cap
  \S_r$, $x_i\in (\S^{loc}_r)^*$, and $y_i\in (\S \setminus
  \S_r)^*$. Then $\state_r(w)=(d_q,x_k)$ and $\state_q(w)=d'_q$, where
  $d_q=\red\state_q(\chi(y_0x_0a_1 \cdots a_k))$ and
  $d'_q=\red\state_q(\chi(y_0x_0a_1 \cdots a_ky_k))$.
  \end{enumerate}
\end{lemma}

\begin{proof}
  Induction on the length of $w=y_0x_0a_1 \cdots a_ky_kx_k$. If $w=\e$
  then $\state_q(\e)=d^1_q$ and $\state_r(\e)=(d^1_q,\e)$ where $d^0_q
  \act{\ch(f_0) \ch(a)} d^1_q$ in $\Dd_q$, which shows the claim. Let
  $w=w'b$. If $b \notin (\S_q\cup \S_r)$, then $x_k=\e$, $y_k=y'b$,
  $\chi(w'b)=\chi(w')b$,
  $\state_q(w'b)=\state_q(w')\stackrel{ind.}{=}\red\state_q(\chi(y_0x_0a_1
  \cdots a_ky'))=\red\state_q(\chi(y_0x_0a_1 \cdots
  a_ky')b)$. Moreover,
  $\state_r(w'b)=\state_r(w')\stackrel{ind.}{=}(d_q,\e)$, where
  $d_q=\red\state_q(\chi(y_0x_0a_1 \cdots a_k))$. Finally, assuming
  that $\red\run(\chi(w'))$ defined, observe that this run can be extended by a
  $b$-transition since it can be in $w$ and the concerned states are
  the same.

We consider the remaining cases:
\begin{enumerate}
\item Let $b \in \Sloc_r$, then $\chi(w'b)=\chi(w')b$ and $x_k=x'b$. We have
  $\state_q(w'b)=\state_q(w')\stackrel{ind.}{=}\red\state_q(\chi(y_0x_0a_1
  \cdots y_k))=: d'_q$.
 Moreover, $\state_r(w')=(d_q,x')$, where
 $d_q=\red\state_q(\chi(y_0x_0a_1 \cdots a_k))$. In
 $\Dd_r$ there is a transition $(d_q,x') \act{b} (d_q,x'b)$, which
 shows the claim about states. Finally we justify that the run on $\chi(w')$ 
 in $\red\Dd$ can be extended by a $b$.  We know that $d_q \act{x'b}$
 and $d'_q \act{x'}$
 in $\red\Dd$, and want to show that $d'_q \act{x'b}$. This holds
 since $\red\Dd$ is covering and since Property (*) guarantees that
 the $s_r$ and $f$ components of $\red\pi(d_q)$ and $\red\pi(d'_q)$ are
 the same.
\item Let $b \in \S_q \setminus \S_r$, so $b$ is either local on $q$
  or a communication with $p \not= q,r$. We have $x_k=\e$ and
  $y_k=y'b$. Assume that $b$ is local on $q$. We have
  $\chi(w)=\chi(w')b \ch(a)$, where $a \in \Sloc_q$ and $d_q$ are such
  that $d_q=\red\state_q(\chi(w'))$ and $d_q \act{b} d^1_q
  \act{\ch(a)}d^2_q$ in $\red\Dd$. By induction,
  $\state_q(w')=\red\state_q(\chi(y_0x_0a_1 \cdots a_ky'))=d_q$, and
  by definition of $\Dd_q$, $d_q \act{b} d^2_q=\ts(d^1_q)$.  Thus
  $\state_q(w)=d^2_q=\red\state_q(\chi(w))$ and the claim about states
  is shown. The run  on $\chi(w)$ in $\red\Dd$ exists by the definition
  of $\chi(w)$ from $\chi(w')$.

The case of a communication with $p\not=r$ is similar to the above.
\item Let $b \in \S_q \cap \S_r$ be a communication between $q$ and
  $r$, thus $a_k=b$ and $x_k=y_k=\e$. We have $\chi(w)=\chi(w') b
  \ch(f) \ch(a)$, where $a,f$ are such that
  $\red\state_q(\chi(w'))=d_q \act{b} d^1_q \act{\ch(f)\ch(a)}
  d^2_q$. Consider $d'_q=\red\state_q(\chi(y_0x_0a_1 \cdots a_{k-1}))$ and
  $d''_q=\red\state_q(\chi(y_0x_0a_1 \cdots a_{k-1}y_{k-1}))$. By
  induction, $\state_q(w')=d''_q$ and $\state_r(w')=(d'_q,x_{k-1})$.
  In $\Dd$ we have a transition $(d''_q,(d'_q,x_{k-1})) \act{b}
  (d^2_q,(d^2_q,\e))$ since $d''_q \act{x_{k-1}} d_q \act{b} d^1_q
  \act{\ch(f)\ch(a)} d^2_q$ in $\red\Dd$. Thus,
  $\state_q(w)=d^2_q=\red\state_q(\chi(w))$ and
  $\state_r(w)=(d^2_q,\e)=(\red\state_q(\chi(w)),\e)$, which shows the
  claim about states. The run on $\chi(w)$ in $\red\Dd$ exists by the
  definition of $\chi(w)$ from $\chi(w')$.
\end{enumerate}
\end{proof}

\begin{lemma}\label{D: maximal}
  If $w \in \slow_r(\Dd)$ and $\run(w)$ is maximal in $\Dd$, then
  $\run(\chi(w))$ is maximal in $\red \Dd$.
\end{lemma}

\begin{proof}
  Recall first that $\run(\chi(w))$ is not maximal only if for some
  finite prefix $x$ of $\chi(w)$, $\run(x)$ can be extended by some
  action $a$ (and the processes in $\loc(a)$ do not appear anymore in
  the remaining suffix of $\chi(w)$). From the definition of $\chi(w)$
  it follows that it suffices to consider prefixes of $\chi(w)$
  of the form $\chi(u)$, where $w=uv$ with $u$ finite. By
  Lemma \ref{lemma:D invariant} we note first that such an $a$ cannot be on
  processes other than $q$ or $r$, since
  $\state_p(u)=\red\state_p(\chi(u))$ for all $p \not= q,r$.

We consider the remaining cases, and assume $u=y_0x_0a_1 \cdots a_k y_kx_k$:

\begin{enumerate}
\item Assume that $\chi(u)$ can be extended by some $b \in\Sloc_r$ in
  $\red\Dd$, and let $d_q=\red\state_q(\chi(y_0x_0a_1 \cdots a_k))$,
  $d'_q=\red\state_q(\chi(y_0x_0a_1 \cdots a_ky_k))$, so $d'_q
  \act{x_kb}$ in $\red\Dd_q$.  By
  Lemma~\ref{lemma:D invariant} we have $\state_r(u)=(d_q,x_k)$ and by
  Property (*), the $s_r$- and $f$-components of $\red\pi(d_q)$ and
  $\red\pi(d'_q)$ are the same. Since $\red\Dd$ is covering, this
  means that $d_q \act{x_kb}$, hence there is a run on $u bv$ in $\Dd$
  so $w$ was not maximal.
\item Assume that $\chi(u)$ can be extended by some $b \in \S_q
  \setminus \S_r$ and recall from Lemma~\ref{lemma:D invariant} that
  $\state_q(u)= d_q$, where $d_q=\red\state_q(\chi(y_0x_0a_1 \cdots
  a_ky_k))$. Consider 
  $u_1=y_0x_0a_1 \cdots a_ky_kb$ and assume that $b$ is $q$-local (the
  case of a communication with $p \not=r$ is similar). We have $d_q
  \act{x_k} d'_q \act{b}$ in $\red\Dd$ from some $d'_q$, and we want
  to show that $d_q 
  \act{b} d''_q$ for some $d''_q$. But this holds since $\red\Dd$ is
  covering and the $s_q$ components of $\red\pi(d_q)$ and
  $\red\pi(d'_q)$ are the same. So the run of $w$ in $\Dd_q$ was not
  maximal, since there is a run on $ubv$ in $\Dd$. 
\item Assume that $\chi(u)$ can be extended by some $b \in \S_q \cap
  \S_r$. Recall from Lemma~\ref{lemma:D invariant} that
  $\state_q(u)= d'_q$ and $\state_r(u)=(d_q, x_k)$, where
  $d_q=\red\state_q(\chi(y_0x_0a_1 \cdots a_k))$ and
  $d'_q=\red\state_q(\chi(y_0x_0a_1 \cdots a_ky_k))$. We have that
  $\red\state_q(\chi(u))=d^1_q$ where $d'_q \act{x_k} d^1_q$, and
  $d^1_q \act{b} d^2_q$. According to the definition of $\Dd$, there is a
  transition $(d'_q,(d_q,x_k)) \act{b} (\ts(d^2_q),(\ts(d^2_q),\e))$
  in $\Dd$, so that the run on $w$ was not maximal.
\end{enumerate}
Note that a run on $\chi(u)$
cannot be extended by actions of the form $\ch(a)$ or $\ch(f)$, since
$\red\Dd$ is covering. So the above four cases exhaust all the possibilities.
\end{proof}

\begin{lemma}\label{lemma:D correct}
 If $\red\Dd$ is a correct covering controller for $\red\Aa$, then $\Dd$
  is a correct covering controller for $\Aa$.
\end{lemma}

\begin{proof}
  By Lemma~\ref{lemma:reduciton-to-slow} it is enough to show that 
  for all $w\in\slow_r(\Dd)$, $\run(w)$ satisfies $\Cor$.
  By Lemmas~\ref{lemma:D invariant} and \ref{D: maximal} the run on
  $\chi(w)$ exists and is 
  maximal. Since $\red\Dd$ is correct this run satisfies $\red\Cor$.

  Consider a maximal run in $\Dd$, labeled by some
  $w\in\slow_r(\Dd)$. It is of one of the forms
  \begin{equation*}
    y_0x_0a_1\cdots a_k y_k x_k a_{k+1}\dots\quad\text{or}\quad
    y_0x_0a_1\cdots a_k x_k y_\omega
  \end{equation*}
  where $a_i\in \S_q \cap \S_r$, $x_i\in (\S^{loc}_r)^*$, $y_i\in (\S
  \setminus \S_r)^*$, and $y_\omega\in (\S \setminus \S_r)^\omega$

  By Lemma~\ref{lemma:D invariant} $\run_p(w)$ and
  $\red\run_p(\chi(w))$ are the same for $p\not=q,r$. Since for such
  $p$ also the
  correctness conditions of $\Aa$ and $\red\Aa$ are the same, and since
  $\red\run_p(\chi(w))$ satisfies $\red\Cor_p$, so does
  $\run_p(w)$. 

  Considering $\run_q(w)$, Lemma~\ref{lemma:D invariant} gives us
  \[\state_q(y_0x_0a_1\cdots a_ky_k)=\state_q(y_0x_0a_1\cdots
  a_ky_kx_k)=\red\state_q(\chi(y_0x_0a_1\cdots a_ky_k))\]
 for every
  $k$. Moreover, the $\Aa_q$-component does not change when going from
  $\red\pi(\chi(y_0x_0a_1\cdots a_ky_k))$ to
  $\red\pi(\chi(y_0x_0a_1\cdots a_ky_kx_k))$. Thus,
  $\pi(\run_q(w))$ is equal to the projection on $\Aa_q$ of $\red\pi(\run_q(\chi(w)))$, so  $\run_q(w)$ satisfies
  $\Cor_q$.

  It remains to consider $\run_r(w)$. For this we can use Lemma~\ref{lemma:D
    invariant} obtaining
 $\state_r(y_0x_0a_1 \cdots a_ky_kx_k)=(d_q,x_k)$ with
  $d_q=\red\state_q(\chi(y_0x_0a_1 \cdots a_k))$, for every
  $k$. Recall that $\pi(d_q,x)$ was defined as the $\Aa_r$-component
  of $\red\pi(d'_q)$, where $d_q \act{x} d'_q$ in $\red\Dd$. Assume
  first that $w$ is of the form $y_0x_0a_1\cdots a_k y_k x_k
  a_{k+1}\dots$. Observe that $\pi(\run_r(w))$ is equal to the
  projection on $\Aa_r$ of $\red\pi(\red\run_q(\chi(w)))$, thus
  $\run_r(w)$ satisfies $\Cor_r$ because $\red\run_r(\chi(w))$ satisfies
  $\Cor_r$. Let now $w$ be of the form $y_0x_0a_1\cdots a_k x_k
  y_\omega$. Since $\run(w)$ is maximal we have that $\state_r(w)\in
  T_r$, again because $\red\run_r(\chi(w))$ satisfies $\Cor_r$. 
\end{proof}

\section{Conclusion}

We have considered a model obtained by instantiating Zielonka automata
into the supervisory control framework of Ramadge and
Wonham~\cite{RW89}. The result is a distributed synthesis framework that
is both expressive and decidable in interesting cases. To
substantiate we have sketched how to encode threaded boolean programs
with compare-and-swap instructions. Our main decidability result
(Theorem~\ref{thm:main}) shows that the synthesis problem is decidable
for hierarchical architectures and for all
local omega-regular specifications. Recall that in the Pnueli and Rosner
setting essentially only pipeline architectures are decidable, with
an additional restriction that only the first and the last process in
the pipeline can handle environment inputs. In our case all the
process can interact with the environment.

The synthesis procedure presented here is in $k$-\EXPTIME\ for
architectures of depth 
$k$, in particular it is \EXPTIME\ for the case of a one server
communicating with clients who do not communicate between each
other. From~\cite{GGMW13} we know that these bounds are tight.

This paper essentially closes the case of tree architectures
introduced in~\cite{GGMW13}. The long standing open question is the
decidability of the synthesis problem for all
architectures~\cite{GLZ04}. 





\newpage

\end{document}